\documentclass[a4paper,11pt,leqno]{amsart}
\usepackage{subfig}
\usepackage{etex}
\usepackage[a4paper,left=25mm,right=25mm,top=30mm,bottom=30mm,marginpar=25mm]{geometry}
\usepackage[utf8x]{inputenc}
\usepackage[english]{babel}
\usepackage{amsmath}
\usepackage{amssymb}
\usepackage{amsthm}
\usepackage{mathabx}
\usepackage{bbm}
\usepackage{enumerate}
\usepackage{empheq}
\usepackage{fancybox}
\usepackage{xcolor}
\usepackage{mathrsfs}
\usepackage{slashed}
\usepackage{tikz}
\usetikzlibrary{shapes,snakes,arrows}
\usetikzlibrary{positioning}
\usepackage{harpoon}
\usepackage{stmaryrd}
\usepackage{setspace}
\usepackage{fancybox}
\usepackage{esint}
\usepackage{bm}
\usepackage{hyperref}
\usepackage{mathtools}

\DeclareRobustCommand{\rchi}{{\mathpalette\irchi\relax}}
\newcommand{\irchi}[2]{\raisebox{\depth}{$#1\chi$}}
\tikzset{node distance=2cm, auto}
\usetikzlibrary{patterns}

\newcommand{\lep}{\lambda \epsilon}
\newcommand{\ol}[1]{\overline{#1}}

\newcommand{\harpoon}{\overset{\rightharpoonup}}
\DeclareMathOperator{\diam}{diam}

\usepackage[all]{xy}
\theoremstyle{definition}

\theoremstyle{plain}
\newtheorem{theorem}{Theorem}
\newtheorem{lemma}{Lemma}
\newtheorem{proposition}[lemma]{Proposition}
\newtheorem{corollary}[lemma]{Corollary}
\theoremstyle{remark}
\newtheorem{remark}{Remark}

\newcommand{\tkpr}[1]{\ensuremath{\tau_{#1}\left(B\left(p, 2R\right)\right)}}

\newcommand{\rB}[1]{\ensuremath{\left(#1\right)}}

\newcommand{\sB}[1]{\ensuremath{\left[#1\right]}}
\newcommand{\cB}[1]{\ensuremath{\left\{#1\right\}}}
\newcommand{\scal}[2]{\ensuremath{\left\langle#1,#2\right\rangle}}

\newcommand{\colvect}[2]{\ensuremath{\left(\begin{matrix} #1 \\ #2 \end{matrix}\right)}}
\newcommand{\rowvect}[2]{\ensuremath{\left(\begin{matrix} #1 & #2 \end{matrix}\right)}}

\newcommand{\norm}[1]{\left|\left|#1\right|\right|}
\newcommand{\modulus}[1]{\left|#1\right|}

\DeclareMathOperator{\Lip}{Lip}

\renewcommand\epsilon{\varepsilon}
\renewcommand\rho{\varrho}

\newcommand{\de}{\ensuremath{\mathrm{d}}}

\newenvironment{frcseries}{\fontfamily{frc}\selectfont}{}

\DeclareMathOperator{\Div}{Div}
\DeclareMathOperator{\Curl}{Curl}
\DeclareMathOperator{\spt}{spt}
\DeclareMathOperator{\dist}{dist}
\DeclareMathOperator{\so}{SO(2)}
\DeclareMathOperator{\id}{Id}

\newcommand{\ta}{\widetilde{A}}
\renewcommand{\tilde}{\widetilde}

\newcommand{\SOn}{SO(n)}

\newcommand{\good}{\ensuremath{\Omega \setminus B_{\lambda \epsilon}\rB{\spt \Curl A}}}
\newcommand{\eel}{\ensuremath{\mathcal{E}_{\text{el}, \epsilon}}}
\newcommand{\ecore}{\ensuremath{\mathcal{E}_{\text{core}, \epsilon}}}
\newcommand{\zak}{%
  \mathbin{\vrule height 1.6ex depth 0pt width
0.13ex\vrule height 0.13ex depth 0pt width 1.3ex}
}

\newcommand{\declareapp}[5]{\ensuremath{\begin{matrix}#1:& #2 &\longrightarrow& #3\\& #4& \longmapsto& #5\end{matrix}}}

\tikzset{node distance=2cm, auto}
\newcommand{\ti}{\tilde{A}}
\renewcommand{\phi}{\varphi}

\DeclareMathOperator{\DEG}{deg}
\newcommand{\mcA}{\mathcal{A}}

\newcommand{\uB}{\bigcup_{i=1}^N B(x_i, \rho_i)}
\renewcommand{\phi}{\varphi}
\newcommand{\points}{\ensuremath{\cB{x_i}_{i=1}^N}}
\newcommand{\tg}{\ensuremath{\tilde{g}}}

\begin{document}
\title[An Energy Estimate for Dislocation Configurations]{An Energy Estimate for Dislocation Configurations and the Emergence of Cosserat-Type Structures in Metal Plasticity}
\author{Gianluca Lauteri}
\address{Max-Planck-Institut f\"ur Mathematik in den Naturwissenschaften, Leipzig, Germany}
\email[G.~Lauteri]{Gianluca.Lauteri@mis.mpg.de}
\author{Stephan Luckhaus}
\address{Institut f\"ur Mathematik, Leipzig University, D-04009 Leipzig, Germany}
\email[S.~Luckhaus]{\tt Stephan.Luckhaus@math.uni-leipzig.de}
\maketitle
\begin{abstract}
We investigate low energy structures of a lattice with dislocations in the context of nonlinear elasticity. We show that these low energy configurations exhibit in the limit a Cosserat-like behavior. Moreover, we give bounds from above and below to the energy of such configurations.
\end{abstract}
\section{Introduction}
We study an energy functional, comparable to the atomistic model introduced in~\cite{LM} and~\cite{LW}, able to describe low energy configurations of a two dimensional lattice with dislocations in a nonlinear elasticity regime. 
Such a model consists of a core energy---which should be viewed as being an energy per atom strictly larger but comparable to the one in the ground state---and a nonlinear elastic energy outside the core, whose size is comparable to the lattice spacing $\epsilon$.\\
The main result can be described as follows: configurations of energy comparable to $\epsilon$ consist of piecewise constant microrotations with small angle grain boundaries between them. 
More precisely, we consider, for admissible strain tensor fields $A$ satisfying appropriate boundary conditions and a topological (but \emph{not} geometrical) constraint and cores $S$, an energy functional of the form
\[
 \mathcal{F}_{\epsilon}(A, S)\coloneqq \mathcal{E}_{\text{el}, \epsilon}(A, S) + \mathcal{E}_{\text{core}, \epsilon}(S),
\]
where $0<\epsilon\ll 1$ is the lattice parameter, $\mathcal{E}_{\text{el}}$ is the elastic energy, which is defined (assuming hyperelasticity outside the core $S$) as an integral functional outside the singular region $S$, and $\mathcal{E}_{\text{core}, \epsilon}$ is the energy of the singular region, which is independent of the strain field.
We can then summarize our results as follows:
\begin{itemize}
	\item The \emph{upper bound}, or the \emph{Read-Shockley formula}, see~\cite{RS} and Theorem~\ref{thm:upper_bound}. That is, we prove in our functional analytic setting the formula which gives the energy of a small angle grain boundary. Namely, for every $\epsilon >0$, we implement a column of dislocations which approximates the grain boundary through an admissible field $A\in\mathcal{A}_{\epsilon}$. More precisely,
     \[
	 \liminf_{\epsilon \downarrow 0} \inf_{(A, S)\text{ admissible}} \frac{1}{\epsilon}\mathcal{F}_{\epsilon}(A, S)  \le C_0 \tau \alpha L \rB{\modulus{\log(\alpha)} + 1}.
     \]
	\item The \emph{compactness in the class of microrotations}, i.e. Theorem~\ref{thm:density} and its Corollary~\ref{cor:microrot} . We say that a matrix field $A \in L^1_{\text{loc}}(\Omega)^{2\times 2}$ is a  \emph{microrotation} if it is a piecewise constant rotation (which can be seen as a generalization of the \emph{tri\`edre mobile} introduced by the Cosserats in~\cite{CC}). Then we can prove that every sequence of admissible pairs $(A_j, S_j)$, whose energy is comparable to the one of a small angle grain boundary, has a \emph{competitor} $(A_j', S_j')$, namely another sequence with ``essentially the same energy'', in the sense that $\mathcal{F}_{\epsilon_j}(A_j', S_j') \le C \mathcal{F}_{\epsilon_j}(A_j, S_j)$ for a \emph{universal} constant $C > 0$, and which moreover is \emph{harmonic} outside the core (see Proposition~\ref{proposition:wlog_harmonic}). We then combine this harmonic competitor with a particular foliation (constructed via an ad hoc covering argument in Lemma~\ref{lemma:foliation}), through a balls construction (inspired to the one used for the Ginzburg-Landau functional, see~\cite{SS} and the references therein). This gives an estimate on the $1$-density with respect to the rescaled energy functional, which we prove in Theorem~\ref{thm:density}. Then, thanks to the geometric rigidity results in~\cite{MSZ} and~\cite{LL}, we see that this competing sequence admits a subsequence which converges strongly in $L^2(\Omega)$ to a microrotation (Corollary~\ref{cor:microrot}).
	\item The \emph{lower bound} to the energy. From the structure of limit fields obtained, we can infer a better lower bound than the one given applying only geometric rigidity. Namely, we prove that the logarithmic term in the upper bound is optimal.
	\end{itemize}
We believe that our result partially explains the microstructure of a metal after the industrial hardening process, which consists of annealing which would lead to low energy configurations and quenching. What remains open is of course to extend the estimate to the three dimensional case and then to the situation of lattices allowing twinned structure.\\
Then there is also the problem of subsequent cold plastic deformations, where cold means that the motion of dislocations is confined to crystallographic glide planes. At the moment we do not even have a conjecture for this situation.\\
It is worthwhile to compare our result with the differential geometric description of dislocation structures, introduced by Kondo, Kr\"oner and Bilby at el. (see~\cite{KO},~\cite{KR},~\cite{BBS}) and also with $\Gamma$-limit results in the context of linear elasticity where implicitly or explicitly a volume density of dislocations is assumed (see~\cite{GLP}). \\
It remains to be investigated if these models remain valid as averaged limits if on an intermediate scale there exists a Cosserat structure of micrograins.
	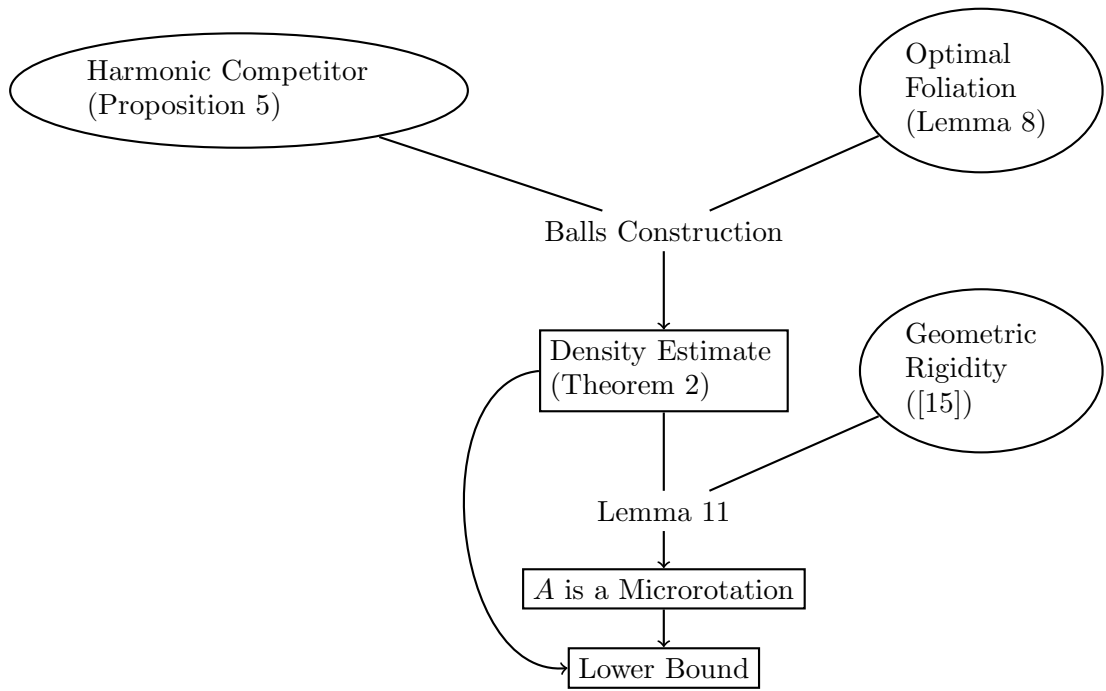
\begin{figure}[!ht]
	 \center
	 \begin{tikzpicture}[node distance=0.5cm, auto]
\matrix [column sep=7mm, row sep=5mm] {
    \node (hc) [draw, shape=ellipse,thick,text width=4cm] {Harmonic Competitor (Proposition~\ref{proposition:wlog_harmonic})}; &
  \node [draw=none, fill=none] {};&
    \node (fl) [draw, shape=ellipse, thick, text width=2cm] {Optimal Foliation (Lemma~\ref{lemma:foliation})}; \\
  \node [draw=none, fill=none] {};&
  \node (bc) [draw=none, fill=none] {Balls Construction}; \\
  \node [draw=none, fill=none] {};&
    \node (de) [draw, shape=rectangle,thick, text width=3cm] {Density Estimate (Theorem~\ref{thm:density})};&
    \node (gr) [draw, shape=ellipse,thick, text width=2cm] {Geometric Rigidity (\cite{MSZ})};\\
  \node [draw=none, fill=none] {};&
    \node (am) [draw=none, fill=none] {Lemma~\ref{lemma:vmeas}};&\\
  \node [draw=none, fill=none] {};&
  \node (mr) [draw, shape=rectangle,thick] {$A$ is a Microrotation};\\
  \node [draw=none, fill=none] {};&
  \node (lb) [draw, shape=rectangle,thick] {Lower Bound};\\
};
\draw[-, thick,bend right] (hc) -- (bc);
\draw[-, thick] (fl) -- (bc);
\draw[->, thick] (bc) -- (de);
\draw[-, thick] (de) -- (am);
\draw[-, thick] (gr) -- (am);
\draw[->,thick] (am) -- (mr);
\draw[->, thick] (mr) -- (lb);
\path[every node/.style={font=\sffamily\small}]
    (de) edge[->, thick, bend right=90] node [left] {} (lb);
\end{tikzpicture}
    \caption{Compactness in the class of Microrotations and Lower Bound.}
    \label{pic:strategy}
	\end{figure}

\section{Notations and Definitions}
 $\SOn$ denotes the group of rotations, i.e.
 \[
  \SOn \coloneqq  \cB{A \in \mathbb{R}^{n\times n}\biggr| A^TA = A A^T = \id,\quad \det(A) = 1}.
 \]
 We also recall that the $\Curl$ of a matrix field $A\in L^1_{\text{loc}}(\Omega)^{2\times 2}$ is defined in the sense of distributions as 
 \footnote{More in general, when $A\in L^1_{\text{loc}}(\Omega)^{n \times n}$, its distributional $\Curl$ can be defined (identifying $A$ with a vector of $1$-forms) as 
	 \[
		 \scal{\Curl A^{(i)}}{\varphi} \coloneqq -\int_{\Omega} \de \varphi \wedge \star A^{(i)},\qquad \forall \varphi \in \mathcal{C}^{\infty}_c\rB{\Omega},
	 \]
	 where $\star$ is the Hodge operator.
 }
 \[
  \scal{\Curl A^{(i)}}{\varphi}\coloneqq \scal{A^{(i)}}{J\nabla \varphi} = \int_{\Omega} A^{(i)} \cdot J\nabla \varphi \de x,\quad J\coloneqq \left(\begin{matrix}0 & -1 \\ 1 & 0\end{matrix}\right),\quad \varphi \in \mathcal{C}^{\infty}_c\rB{\Omega},
 \]
 while the support of a distribution $T \in \mathcal{D}'(\Omega)$ is
 \[
  \spt T \coloneqq  \Omega\setminus \rB{\bigcup_{\substack{U \subset \Omega \text{ open }\\T(\varphi) = 0 \:\forall \varphi\in \mathcal{C}^{\infty}_c(U)}}{U}}.
 \]
 With $\mathcal{L}^n$ or $\modulus{\cdot}$ we denote the Lebesgue measure on $\mathbb{R}^n$, while we write $\mathcal{H}^k$ for the $k$-dimensional Hausdorff measure. With $C>0$ we always denote constants which depend at most on the dimension (i.e., $2$ in our analysis), and may vary from line to line.\\
 We can now state the problem. In what follows
 \begin{itemize}
  \item $\Omega\coloneqq [-L, L]^2$ represents a section of a crystal, $L>0$;
  \item $\epsilon > 0$ is the \emph{lattice parameter}, i.e. the distance between atoms;
  \item $1 \gg \alpha > 0$ is the (``small'') misorientation angle between two grains;
  \item $\ell > 0$ is a parameter (much smaller than $L$): in an $\ell$-neighborhood of $x = \pm L$ we are going to impose the boundary conditions;
  \item $\lambda > 0$ is a parameter (independent of $L, \epsilon, \alpha$) so that $\lambda \epsilon$ gives what physicists call the \emph{core radius};
  \item $\tau > 0$ is another parameter independent of all the others, which is defining the \emph{minimal length of the Burgers' vector}, $\tau \epsilon$.
 \end{itemize}
We then restrict our attention to the following class of \emph{admissible strain fields}, denoted by $\mathcal{A}\rB{\epsilon, \alpha, L, \tau, \lambda, \ell}$ (to which we shall simply refer to as $\mathcal{A}_{\epsilon}$, in the case when the other parameters are clear from the context), which is defined as the family of matrix fields $A: \Omega \to \mathbb{R}^{2\times 2}$ satisfying the following conditions:
\begin{enumerate}[($\mathcal{A}_{\epsilon}$, i)]
  \item $A \in L^1_{\text{loc}}(\Omega)^{2\times 2}$ and $A \in L^2\rB{\good}^{2\times 2}$;
  \item (\emph{Boundary Condition}) 
  $A \equiv R_{\alpha}$ in $[-L, -L+\ell]\times [-L, L]$ and $A\equiv R_{-\alpha}$ in $[L-\ell, L]\times [-L, L]$, where $R_{\alpha}$ is the counter-clockwise rotation through the angle $\alpha$, that is
  \[
   R_{\alpha} = \sB{\begin{matrix}\cos(\alpha)&-\sin(\alpha)\\\sin(\alpha)&\cos(\alpha)\end{matrix}};
  \]
  \item (\emph{First Quantization of the Burgers' vector}) For every closed, Lipschitz simple curve $\gamma \subset \Omega\setminus B_{\lambda \epsilon}(\spt\Curl(A))$, either
  \[
      \int_{\gamma}{A\cdot t \mathrm{d}\mathscr{H}^1} = 0
  \]
  or
  \[
      \modulus{\int_{\gamma}{A\cdot t \mathrm{d}\mathscr{H}^1}}\ge \tau \epsilon.
  \]
\end{enumerate}
We call an \emph{admissible core} any compact subset of $[-L+\ell, L-\ell]\times [-L, L]$, i.e. any element of $\mathcal{K}([-L+\ell, L-\ell]\times [-L, L])$.
The \emph{elastic energy} of a pair $(A, S) \in \mathcal{A}_{\epsilon} \times \mathcal{K}([-L+\ell, L-\ell]\times [-L, L])$ is 
\[
 \eel(A, S)\coloneqq \frac{1}{\tau}\int_{\Omega\setminus B_{\lambda\epsilon}(S)} \dist^2(A, \so) \de x,
\]
while the \emph{core energy} depends only on the core and is defined as
\[
 \ecore(S)\coloneqq \frac{1}{\lambda^2} \modulus{B_{\lambda\epsilon}(S)}.
\]
We define the set of admissible pairs
\[
 \mathcal{P}\rB{\epsilon, \alpha, L, \tau, \lambda, \ell}\coloneqq \mathcal{A}\rB{\epsilon, \alpha, L, \tau, \lambda, \ell}\times \mathcal{K}([-L+\ell, L-\ell]\times [-L, L])
\]
Whenever the constants $\alpha, L, \tau, \lambda, \ell$ are clear from the context, we shall simply write $\mathcal{P}_{\epsilon}$ for $\mathcal{P}\rB{\epsilon, \alpha, L, \tau, \lambda, \ell}$.
The (free) energy functional is defined on pairs $(A, S) \in \mathcal{P}_{\epsilon}$ as
\[
 \mathcal{F}_{\epsilon}(A, S)\coloneqq \begin{cases}
 								\eel(A, S) + \ecore(S)&\text{if }\spt(\Curl(A))\subset S,\\
 								+\infty&\text{otherwise}.
 							   \end{cases}
\]
We also define the relaxed energy on admissible fields as
\[
 \mathcal{F}_{\epsilon}(A)\coloneqq \mathcal{F}_{\epsilon}(A, \spt(\Curl(A))).
\]
For notational simplicity, for a set $S$ we let $\Omega_{\lambda\epsilon}(S)\coloneqq \Omega\setminus B_{\lambda\epsilon}(S)$, while for a matrix field $A$ $\Omega_{\lambda\epsilon}(A)\coloneqq \Omega\setminus B_{\lambda\epsilon}(\spt(\Curl(A)))$.\\
Recall that a function $u \in L^1(\Omega)$ is in $BV(\Omega)$ if its distributional derivative $Du$ is a finite Radon measure. Moreover, the derivative can be written as
\[
 Du = \nabla u \mathcal{L}^n + \rB{u^+ - u^-} \nu_u \mathcal{H}^{n-1}\zak S_u + D^c u,
\]
where $\nabla u$ is the approximate gradient, $\nu_u$ is the unit normal to the singular set $S_u$ of $u$ and $D^c u$ is the Cantor part of the derivative (we refer to~\cite{AFP} for more details).\\
We say that a matrix field $A$ is a \emph{microrotation} if $A \in BV(\Omega)^{n\times n}$, $A(x) \in \SOn$ for almost every $x \in \Omega$ and $DA = D^J A$, i.e.
\[
 DA^{(i)} = \rB{A^{(i),+} - A^{(i),-}} \otimes \nu_{A} \mathcal{H}^{n-1}\zak S_A,\qquad i=1, \cdots, n.
\]
 Recall that is well defined a trace for matrix fields whose $\Curl$ is square-integrable, in the following sense. If $U$ is a bounded Lipschitz domain in $\mathbb{R}^2$,
 \[
  H(\Curl, U)^{2\times 2}\coloneqq \cB{A \in L^2(U)^{2 \times 2}\biggr| \Curl(A) \in L^2\rB{U}}.
 \]
 Then, the operator
 \[
  \gamma: A \in H(\Curl, U) \longmapsto A\cdot t \in H^{-\frac{1}{2}}(\partial U),
 \]
 is well defined and continuous (where $t(x)$ is the tangent vector to $\partial U$ at the point $x$), in particular there exists a constant $C = C(U)>0$ such that
 \[
  \modulus{\int_{\partial U}{A \cdot t \de \mathcal{H}^1}} \le C \norm{A}_{H(\Curl, U)}.
 \]
 Moreover, an approximation argument (see~\cite{DL}) gives
 \[
  \int_{\partial U} A\cdot t \de \mathcal{H}^1 = \int_{U}\Curl(A) \de x \qquad \forall A \in H\rB{\Curl, \Omega}^{2\times 2}.
 \]
 To every $\gamma$ closed, Lipschitz, simple curve contained in $\Omega_{\lambda \epsilon}(A)$ we associate its \emph{Burgers' vector} defined as
\[
 \harpoon{b}(\gamma)\coloneqq \int_{\gamma} A \cdot t \de \mathcal{H}^1.
\]
\begin{remark}
 Although we chose $\dist^2(\cdot, SO(2))$ as the elastic energy density, all the results we prove remain valid if we consider instead a function $W: \mathbb{R}^{2\times 2}\to [0, \infty)$ which satisfies the usual assumptions of an elastic energy density in (two dimensional) nonlinear elasticity, that is
 \begin{enumerate}[(i)]
  \item $W$ is continuous and of class $\mathcal{C}^2$ in a neighborhood of $SO(2)$;
  \item $W(\text{Id}) = 0$, i.e. the reference configuration is stress-free;
  \item $W(RA) = W(A)$ for every matrix $A \in \mathbb{R}^{2\times 2}$, i.e. $W$ is frame indifferent,
 \end{enumerate}
 together with the growth assumption
 \begin{enumerate}[(iv)]
  \item There exists a constant $C > 1$ such that $C^{-1}\dist^2(A, SO(2)) \le W(A) \le C\dist^2(A, SO(2))$.
 \end{enumerate}
 Condition (iv) is rather restrictive, but it is essential in order to apply the Geometric Rigidity estimate of M\"uller, Scardia and Zeppieri (~\cite{MSZ}).
\end{remark}
\section{The Read-Shockley Formula}
\begin{theorem}
 \label{thm:upper_bound}
 There exists a constant $C_0>0$ such that
 \begin{equation}
  \label{eq:RS}
  \liminf_{\epsilon \downarrow 0} \inf_{(A, S) \in \mathcal{P}_{\epsilon}} \frac{1}{\epsilon}\mathcal{F}_{\epsilon}(A, S)  \le C_0 \tau \alpha L \rB{\modulus{\log(\alpha)} + 1}.
 \end{equation}
\end{theorem}
\begin{proof}
Consider $\ol{n} \in \mathbb{N}$ such that $\frac{1}{\alpha} \in [2^{\ol{n}}, 2^{\ol{n}+1})$. Without loss of generality, we can assume $\epsilon \in L2^{1-k} \frac{1}{2\mathbb{N}}\coloneqq \cB{\frac{L2^{1-k}}{2z}\biggr| z \in \mathbb{N}}$. Set $r_0\coloneqq \frac{\lambda\epsilon}{2}$ and $N\coloneqq \frac{L}{2^k r_0} \in 2\mathbb{N}$. Let $r_n\coloneqq 2^n r_0$ and
 \[
  \begin{split}
   p^1_n&\coloneqq \rB{-r_n, r_n},\: p^2_n\coloneqq \rB{r_n, r_n},\: p^3_n\coloneqq \rB{r_n, -r_n},\: p^4_n\coloneqq (-r_n, -r_n) \text{ for } n = 0, \cdots, \ol{n},\\
   \Delta^1_n&\coloneqq \Delta\rB{p^1_n, p^1_{n-1}, p^4_{n-1}},\: \Delta^2_n\coloneqq \Delta\rB{p^1_n, p^4_{n-1}, p^4_n},\\
   \Delta^3_n&\coloneqq \Delta\rB{p^2_n, p^3_{n-1}, p^3_n},\: \Delta^4_n\coloneqq \Delta\rB{p^2_n, p^2_{n-1}, p^3_{n-1}} \text{ for }n = 1, \cdots, \ol{n},
  \end{split}
 \]
 where $\Delta(a, b, c)$ denotes the triangle whose vertices are $a$, $b$ and $c$.
\begin{figure}[!ht]
  \center
  \begin{tikzpicture}[scale=0.7]
\draw [fill=black, ultra thick] (-1,-1) rectangle (1, 1);
     \draw [pattern=north west lines, pattern color = gray] (-2, 2) -- (-1, 1) -- (-1, -1) -- (-2, 2);
 \draw [pattern=north east lines, pattern color=gray] (-2, 2) -- (-2, -2) -- (-1, -1) -- (-2, 2);
 \draw [pattern=north west lines, pattern color=gray] (2, 2) -- (1, 1) -- (1, -1) -- (2, 2);
 \draw [pattern=north east lines, pattern color=gray] (2, 2) -- (1, -1) -- (2, -2) -- (2, 2);
    \draw [fill=black, ultra thick] (-1,-1) rectangle (1, 1);
    \draw [ultra thick] (-2,-2) rectangle (2, 2);
 \draw [thin] (-2, -2) -- (-1, -1);
 \draw [thin] (-2, 2) -- (-1, 1);
 \draw [thin] (2, 2) -- (1, 1);
 \draw [thin] (2, -2) -- (1, -1);
 \draw [thin] (-2, 2) -- (-1, -1);
 \draw [thin] (2, 2) -- (1, -1);
 \draw [thin] (0, 1) -- (0,2);
 \node [scale=0.6] at (0, -1.5) {$id$};
 \node [scale=0.6] at (-0.6, 1.5) {$id+ \harpoon b$};
 \node [scale=0.6] at (0.6, 1.5) {$id- \harpoon b$};
 \draw [thin] (-4, -4) -- (-2, -2);
 \draw [thin] (-4, 4) -- (-2,2);
 \draw [thin] (4, 4) -- (2, 2);
 \draw [thin] (4, -4) -- (2, -2);
 \draw [thin] (-4, 4) -- (-2, -2);
 \draw [thin] (4, 4) -- (2, -2);
 \draw [thin] (0, 2) -- (0,4);
 \node at (0, -3) {$id$};
 \node at (-1.2, 3) {$id+ \harpoon b$};
 \node at (1.2, 3) {$id- \harpoon b$};
 \draw [pattern=north west lines, pattern color = gray] (-4, 4) -- (-2, 2) -- (-2, -2) -- (-4, 4);
 \draw [pattern=north east lines, pattern color=gray] (-4, 4) -- (-4, -4) -- (-2, -2) -- (-4, 4);
 \draw [pattern=north west lines, pattern color=gray] (4, 4) -- (2, 2) -- (2, -2) -- (4, 4);
 \draw [pattern=north east lines, pattern color=gray] (4, 4) -- (2, -2) -- (4, -4) -- (4, 4);
    \draw [ultra thick] (-4,-4) rectangle (4,4);
\end{tikzpicture}
    \caption{The map $v^{(1)}$ (the striped triangles are the ones where we are interpolating).}
   \label{pic:v_1}
\end{figure}
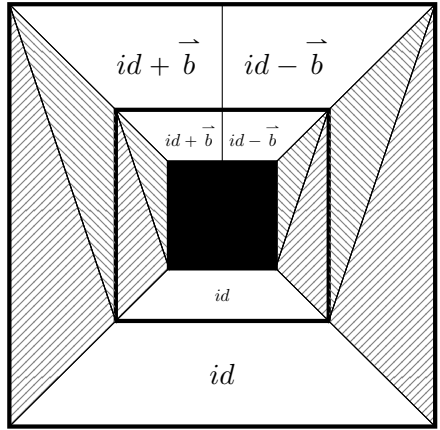

 Let $Q_n\coloneqq \sB{-r_n, r_n}^2$ and $Q\coloneqq \sB{-r_{\ol{n}}, r_{\ol{n}}}^2$, $\harpoon b\coloneqq \rB{\tau \epsilon, 0}$ and define (see Figure~\ref{pic:v_1})
 \[
  v^{(1)}\coloneqq \begin{cases}
                  \text{id}&\text{in } \rB{\rB{Q \setminus \bigcup_{i, n}\Delta^i_n} \cap \cB{y < 0}}\cup \sB{-r_0, r_0}^2,\\
                  \text{id} + \harpoon b&\text{in } \rB{Q \setminus \bigcup_{i, n}\Delta^i_n} \cap \cB{y > 0, x \le 0},\\
                  \text{id} - \harpoon b&\text{in } \rB{Q \setminus \bigcup_{i, n}\Delta^i_n} \cap \cB{y > 0, x > 0},\\
                  \text{linear interpolation}&\text{ in } \bigcup_{i, n}\Delta^i_n.
                 \end{cases}
 \]
 \begin{figure}[!ht]
  \center
  \begin{tikzpicture}[scale=1]
  \pgfmathsetmacro{\b}{0.3}
  \node at (-1.3, 0) {$R_{-\alpha}$};
  \node at (1.3, 0) {$R_{\alpha}$};
  \draw[pattern=north west lines, pattern color=gray] (-2+\b, 2) -- (0, 2) -- (-1+\b, 1) -- (-2+\b, 2);
  \draw[pattern=north east lines, pattern color=gray] (0, 2) -- (-1+\b, 1) -- (1-\b, 1)--(0, 2);
  \draw[pattern=north west lines, pattern color = gray] (2-\b, 2) -- (0, 2) -- (1-\b, 1) -- (2-\b, 2);
  \draw[pattern=north east lines, pattern color=gray] (-2, -2) -- (-1, -1) -- (0, -2) -- (-2, -2);
  \draw[pattern=north west lines, pattern color=gray] (0, -2) -- (-1, -1) -- (1, -1) -- (0, -2);
  \draw[pattern=north east lines, pattern color=gray] (2, -2) -- (0, -2) -- (1, -1) -- (2, -2);
  \draw [fill=black] (-1+\b, 1) -- (1-\b, 1) -- (1, -1) -- (-1, -1) -- (-1+\b, 1);
  \draw [ultra thick] (-2+\b, 2) -- (2 -\b, 2) -- (2, -2) -- (-2, -2) -- (-2+\b, 2);
  \draw [thin] (-2+\b, 2) -- (-1+\b, 1);
  \draw [thin] (-2, -2) -- (-1,-1);
  \draw [thin] (2, -2) -- (1, - 1);
  \draw [thin] (1 - \b, 1) -- (2-\b, 2);
  \draw [thin] (-1+\b, 1) -- (0, 2);
  \draw [thin] (0, 2) -- (1-\b, 1);
  \draw [thin] (-1, -1) -- (0, -2);
  \draw [thin] (0, -2) -- (1, -1);
 \draw [fill] (0, -2) circle [radius=0.1];
 \draw [fill] (0, 2) circle [radius=0.1];
 \node at (0, 2.3) {$id$};
 \node at (0, -2.3) {$id$};
\end{tikzpicture}
 \caption{The map $v^{(2)}$ (as in Figure~\ref{pic:v_1}, the stripes denote the regions where we are interpolating).}
 \label{pic:v_2}
  \end{figure}
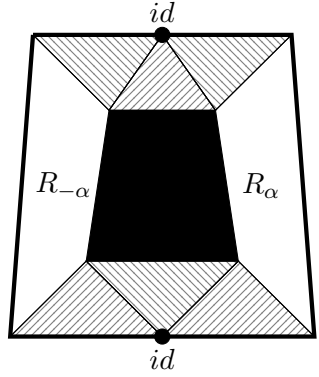
 It is readily seen that for $p \in \Delta^i_n$ we have
 \[
  \modulus{\nabla v^{(1)}(p) - \text{id}} \le C \frac{1}{2^n}. 
 \]
 Now, we have to adjust the boundary condition. For, we consider the map $v^{(2)}: v^{(1)}\rB{Q} \to \mathbb{R}^2$ defined as follows (see Figure~\ref{pic:v_2}). For $n = 1, \cdots, \ol{n}$, define the points
 \[
  \begin{split}
   &q^{1}_n\coloneqq (r_n - \epsilon, r_n),\qquad q^2_n\coloneqq (r_n, -r_n),\qquad q^3_n\coloneqq (0, -r_n),\\
   &q^4_n\coloneqq (-r_n, -r_n),\qquad q^5_n\coloneqq (-r_n + \epsilon, r_n),\qquad q^6_n\coloneqq (0, r_n),
  \end{split}
 \]
 and
 \[
  \begin{split}
   &q^1_0\coloneqq (r_0, r_0),\qquad q^2_0\coloneqq (r_0, -r_0),\qquad q^3_0\coloneqq (0, -r_0),\\
   &q^4_0\coloneqq (-r_0, -r_0),\qquad q^5_0\coloneqq (-r_0, r_0),\qquad q^6_0\coloneqq (0, r_0).
  \end{split}
 \]
 Then, for $n=0, \cdots, \ol{n}$, consider the triangles
 \[
  \begin{split}
   &\tilde{\Delta}^1_n\coloneqq \Delta\rB{q^5_n, q^5_{n-1}, q^6_{n-1}},\qquad \tilde{\Delta}^2_n\coloneqq \Delta\rB{q^5_{n-1}, q^6_n, q^1_{n-1}},\qquad \tilde{\Delta}^3_n\coloneqq \Delta\rB{q^6_n, q^1_{n-1}, q^1_n},\\
   &\tilde{\Delta}^4_n\coloneqq \Delta\rB{q^2_n, q^2_{n-1}, q^3_n},\qquad \tilde{\Delta}^5_n\coloneqq \Delta\rB{q^2_{n-1},q^3_n, q^4_{n-1}},\qquad \tilde{\Delta}^6_n\coloneqq \Delta\rB{q^3_n, q^4_{n-1},q^4_n}.
  \end{split}
 \]
 We then define $v^{(2)}:v^{(1)}(Q)\to \mathbb{R}^2$ as
 \[
  v^{(2)}(x)\coloneqq \begin{cases}
  				R_{-\alpha} x&\text{if } x \in \bigcup_{n=1}^{\ol{n}}\cB{v^{(1)}\rB{Q_n\setminus Q_{n-1}} \setminus \bigcup_{j=1}^6 \tilde{\Delta}^i_n} \cap \cB{x < 0},\\
  				R_{\alpha} x&\text{if }x \in \bigcup_{n=1}^{\ol{n}}\cB{v^{(1)}\rB{Q_n\setminus Q_{n-1}} \setminus \bigcup_{j=1}^6 \tilde{\Delta}^i_n} \cap \cB{x > 0},\\
  				x&\text{if }x=\rB{0, \pm r_n},\\
  				\text{linear interpolation}&\text{otherwise}.
  			  \end{cases}
 \]
 It is easy to check that on each triangle we have
  \[
   \dist^2\rB{\nabla v^{(2)}, \so} \le C \alpha^2.
  \]
\begin{figure}[!ht]
   \center
   \begin{tikzpicture}[scale=0.5]
   \pgfmathsetmacro{\T}{4}
   \draw [fill=black] (-1, -1+\T) rectangle (1, 1+\T);
 \draw [ultra thick] (-2, -2+\T) rectangle (2, 2+\T);
 \draw [fill=black] (-1, -1) rectangle (1, 1);
 \draw [ultra thick] (-2, -2) rectangle (2, 2);
 \draw [fill=black] (-1, -1-\T) rectangle (1, 1-\T);
 \draw [ultra thick] (-2, -2-\T) rectangle (2, 2-\T);
 \draw [ultra thick, ->] (2.5, 0) -- (4.5, 0);
  \node at (3.5, 0.5) {$v^{(1)}$};
 \pgfmathsetmacro{\t}{7}
 \pgfmathsetmacro{\b}{0.3}
 \draw [fill=black] (-1+\t+\b, 1+\T) -- (1+\t-\b, 1+\T) -- (1+\t, -1+\T) -- (-1+\t, -1+\T) -- (-1+\t+\b, 1+\T);
 \draw [ultra thick] (-2+\t+\b, 2+\T) -- (2 + \t-\b, 2+\T) -- (2+\t, -2+\T) -- (-2 + \t, -2+\T) -- (-2+\t+\b, 2+\T);
 \draw [fill=black] (-1+\t+\b, 1) -- (1+\t-\b, 1) -- (1+\t, -1) -- (-1+\t, -1) -- (-1+\t+\b, 1);
 \draw [ultra thick] (-2+\t+\b, 2) -- (2 + \t-\b, 2) -- (2+\t, -2) -- (-2 + \t, -2) -- (-2+\t+\b, 2);
 \draw [fill=black] (-1+\t+\b, 1-\T) -- (1+\t-\b, 1-\T) -- (1+\t, -1-\T) -- (-1+\t, -1-\T) -- (-1+\t+\b, 1-\T);
 \draw [ultra thick] (-2+\t+\b, 2-\T) -- (2 + \t-\b, 2-\T) -- (2+\t, -2-\T) -- (-2 + \t, -2-\T) -- (-2+\t+\b, 2-\T);
 \draw [ultra thick, ->] (2.5 +\t, 0) -- (4.5 + \t, 0);
 \node at (3.5+\t, 0.5) {$v^{(2)}$};
 \draw [fill=black, cm={1,0,0,1,(2*\t,0)}] (-0.5, 1.1+\T)--(-1.1, -0.8+\T)--(0, -1+\T)--(1.1, -0.8+\T)--(0.5, 1.1+\T)--(0, 1+\T)--(-0.5, 1.1+\T);
 \draw [ultra thick, cm={1, 0, 0, 1, (2*\t, 0)}] (-1.4, 2.2+\T) -- (-2.3, -1.7+\T) -- (0, -2+\T) -- (2.3, -1.7+\T) -- (1.4, 2.2+\T) -- (0, 2+\T) -- (-1.4, 2.2+\T);
 \draw [fill=black, cm={1,0,0,1,(2*\t,0)}] (-0.5, 1.1)--(-1.1, -0.8)--(0, -1)--(1.1, -0.8)--(0.5, 1.1)--(0, 1)--(-0.5, 1.1);
 \draw [ultra thick, cm={1, 0, 0, 1, (2*\t, 0)}] (-1.4, 2.2) -- (-2.3, -1.7) -- (0, -2) -- (2.3, -1.7) -- (1.4, 2.2) --  (0, 2) -- (-1.4, 2.2);
 \draw [fill=black, cm={1,0,0,1,(2*\t,0)}] (-0.5, 1.1-\T)--(-1.1, -0.8-\T)--(0, -1-\T)--(1.1, -0.8-\T)--(0.5, 1.1-\T)--(0, 1-\T)--(-0.5, 1.1-\T);
 \draw [ultra thick, cm={1, 0, 0, 1, (2*\t, 0)}] (-1.4, 2.2-\T) -- (-2.3, -1.7-\T) -- (0, -2-\T) -- (2.3, -1.7-\T) -- (1.4, 2.2-\T) -- (0, 2-\T) -- (-1.4, 2.2-\T);
\end{tikzpicture}
   \caption{Schematic representation of the grain boundary constructed.}
   \label{pic:v}
\end{figure}
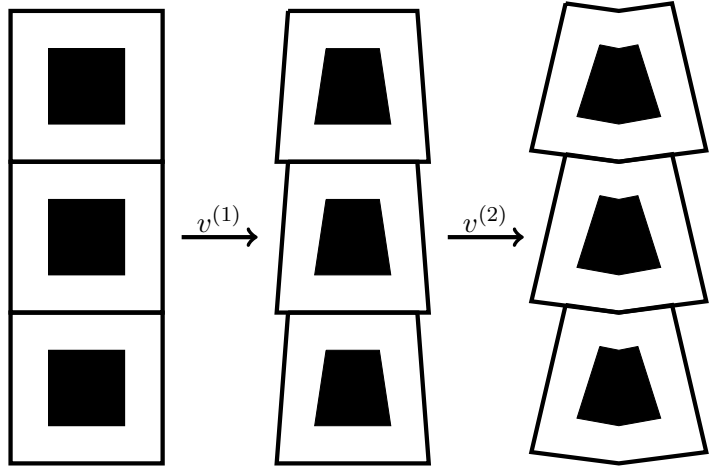
$\:$\\
Thus, if $v\coloneqq  v^{(2)}\circ v^{(1)}$ (see Figure~\ref{pic:v}), on each triangle $\Delta^i_n$,
  \[
   \dist^2\rB{\nabla v, \so} \le C \frac{1}{4^n} + \alpha^2.
  \]
This gives in particular
  \[
   \int_{Q\setminus \sB{-r_0, r_0}^2}{\dist^2\rB{\nabla v, \so}\de x} \le C\rB{ \epsilon^2 \modulus{\log(\alpha)} + \epsilon^2}.
  \]
The last step consists in gluing the maps constructed before. Namely, if $S\coloneqq [-r_{\ol{n}}, r_{\ol{n}}]\times [-L, L]$ we define the map $u: \Omega \to \mathbb{R}^2$ as
  \[
   u(x, y)\coloneqq \begin{cases}
             R_{-\alpha}\colvect{x}{y}&\text{if }(x, y) \in \rB{\Omega\setminus S}\cap \cB{x < 0},\\
             R_{\alpha}\colvect{x}{y} &\text{if }(x, y) \in \rB{\Omega\setminus S}\cap \cB{x > 0},\\
             v\rB{x, y + kr_{\ol{n}}}& \text{if }(x, y) \in Q + k\rB{0, r_{\ol{n}}},\quad k \in \cB{-\frac{N}{2}, \frac{N}{2}}.
            \end{cases}
  \]
  Then, if $A_{\text{gb}}\coloneqq \nabla u$,
  \[
   \mathcal{F}(A_{\text{gb}}) \le C \tau \epsilon \alpha h\rB{\modulus{\log \alpha} + 1}.\qedhere
  \]
  \end{proof}
  We say that $E_{\text{gb}}(\epsilon)\coloneqq  C_0 \tau \epsilon \alpha L \rB{\modulus{\log(\alpha)} + 1}$ is the \emph{energy of a grain boundary with misorientation angle} $\alpha$ \emph{at the scale} $\epsilon$, where $C_0 > 0$ is the constant from Theorem~\ref{thm:upper_bound}.
  \section{Structure of limit fields and lower bound}
  We start this section with a couple of technical lemmas. The first one ensures, through a Whitney-like extension, the existence of competitors which are uniformly bounded by a universal constant and smooth outside their singular set, while the second one allows to find a competitor whose total variation of the $\Curl$ is controlled by the core energy.
 \begin{lemma}
  \label{lemma:wlog_bdd}
  There exists a constant $C > 0$ such that for every pair $(A, S) \in \mathcal{P}\rB{\epsilon, \alpha, L, \ell \lambda}$  whose energy satisfies $\mathcal{F}_{\epsilon}(A, S) \le E_{\text{gb}}(\epsilon)$ there exists another pair $(\tilde{A}, \tilde{S}) \in \mathcal{P}\rB{\epsilon, \alpha, L, \frac{\ell}{2}, \lambda}$ s.t.
  \begin{enumerate}[(i)]
   \item $\norm{\tilde{A}}_{L^{\infty}(\Omega)} \le C$;
   \item $\mathcal{F}_{\epsilon}(\tilde{A}, \tilde{S}) \le C \mathcal{F}_{\epsilon}(A, S)$;
   \item $\tilde{A} \in \mathcal{C}^{\infty}\rB{\overline{\Omega\setminus B_{\lambda \epsilon}(\tilde{S})}}$.
  \end{enumerate}
 \end{lemma}
 \begin{proof}
  We will define the pair $(\tilde{A}, \tilde{S})$ by modifying it in several steps. Let $\omega\coloneqq  B_{\lambda \epsilon}(S)$ and define
  \[
   \tilde{A}_1\coloneqq \begin{cases}
                  \text{id}&\text{in }\omega,\\
                  A&\text{in }\Omega\setminus \omega.
                \end{cases}
  \]
  Clearly, $\spt \Curl \ti_1 \subseteq \omega=:\tilde{S}_1$, and by Vitali's Lemma we can find an at most countable collection of point $\rB{x_j}_{j \in J} \in S$ such that the balls $B_{\lep}(x_j)$ are mutually disjoint and
  \[
   \omega = \bigcup_{x \in S} B_{\lep} (x) \subset \bigcup_{j \in J} B_{5 \lep}(x_j).
  \]
  Thus
  \[
   \modulus{B_{\lep}(\tilde{S}_1)} \le \modulus{B_{\lep}\rB{\omega}} \le \modulus{\bigcup_{j\in J} B_{6\lep}(x_j)} \le C_{\lambda}\sum_{j \in J} {\modulus{B_{\lep}(x_j)}} \le C_{\lambda, n}\modulus{\omega} \le C_{\lambda} \mathcal{F}_{\epsilon}(A, S).
  \]
  Thus $\mathcal{F}_{\epsilon}(\ti_1, \tilde{S}_1) \le C_{\lambda} \mathcal{F}_{\epsilon, S}$ and $\norm{\ti_1}_{L^{\infty}(\omega)} \le M$. For notational simplicity, relabel $\ti_1$ as $A$ and $\tilde{S}_1$ as $S$.
  Now we show that we can without loss of generality assume $A$ to be smooth outside $B_{\lambda \epsilon}(S)$. By the Hodge-Morrey decomposition, $A = \nabla u + F$, where $u \in W^{1, 2}(\Omega)$ and $F \in L^2(\Omega)$ has zero divergence in the sense of distributions. Moreover, $\Curl(F) = 0$ in $\Omega_{\frac{\lambda\epsilon}{2}}(S)$, and hence is harmonic (and, in particular, smooth) in $\Omega_{\frac{\lambda\epsilon}{2}}(S)$. We then take a sequence $u_k\in\mathcal{C}^{\infty}(\Omega)\cap W^{1, 2}(\Omega)$ converging to $u$ in $W^{1,2}(\Omega)$. Set $A_k \coloneqq  \nabla u_k + F$. Clearly $\Curl(A_k) = \Curl(F) = \Curl(A)$ in $\Omega$ for every $k$ and 
  \[
   \begin{split}
   \int_{\Omega}{\dist^2(A_k, \so)\de x} &\le 2\rB{\int_{\Omega}\dist^2(A, \so)\de x + \int_{\Omega}\modulus{\nabla (u_k - u)}^2\de x} \\
   &\le 3\int_{\Omega}\dist^2(A, \so)\de x,
   \end{split}
  \]
  provided $k$ is chosen big enough. That is, we can without loss of generality assume $A$ to be smooth in $\overline{\Omega_{\lambda\epsilon}(S)}$. Now, fix $M > 1$ and consider the set of points
  \[
   R\coloneqq  R_M\coloneqq \cB{x \in \Omega\biggr|\exists r > 0 \text{ : }\fint_{B_r(x)} \dist^2(A, \so) \ge M},
  \]
  and define
  \[
   r(x)\coloneqq  r_M(x)\coloneqq \inf\cB{r>0\biggr|\fint_{B_r(x)} \dist^2(A, \so) \ge M}.
  \]
  Clearly, $\norm{A}_{L^{\infty}(\Omega\setminus R)} \le M + 2\sqrt{n}$. Let $R_1\coloneqq  R\cap \cB{r(x) \ge \epsilon}$, and define the new field
  \[
   \ti_2\coloneqq \begin{cases}
            \text{id}&\text{in } B_1,\\
            A&\text{in }\Omega\setminus B_1,
          \end{cases}
  \]
  where $B_1\coloneqq \bigcup_{x \in R_1}B_{r(x)}(x)$. Then $\spt\Curl \ti_2 \subset B_1 \cup S$. Set $\tilde{S}_2\coloneqq  S\cup R_1$. Using Vitali's Lemma as before, we find a collection of (at most countable) mutually disjoint balls $B_j = B_{r(x_j)}(x_j)$ whose centers are in $R_1$ and
  \[
   R_1 \subset \bigcup_{j \in J} B_{5r(x_j)}(x_j).
  \]
  Thus, since $r(x_j)=:r_j\ge \epsilon$ for every $j\in J$,
  \[
   \modulus{B_{\lep}(R_1)} \le \sum_{j\in J}{\modulus{B_{(5+\lambda)r_j}(x_j)}} \le C_{\lambda}\sum_j \modulus{B_j} \le \frac{C_{\lambda}}{M} \sum_{j \in J} \int_{B_j}\dist^2(A, \so) \le \frac{C_{\lambda}}{M} \mathcal{F}_{\epsilon}(A, S).
  \]
  As done before, relabel for simplicity $\ti_2$ as $A$ and $\tilde{S}_2$ as $S$, and redefine the set $R$ and the function $r$ in function of this new pair $(A, S)$. Then we reduced ourselves to the case when the potentially bad points, i.e. the ones in $R$, have $r(x) < \epsilon$. Consider first those which lie in $B_{\lep}\rB{S}$, i.e. the points in $R_2\coloneqq R\cap B_{\lep}\rB{S}$. Consider the field
  \[
  \ti_3\coloneqq \begin{cases}
            \text{id}&\text{in } B_2,\\
            A&\text{in }\Omega\setminus B_2,
          \end{cases}
  \]
  and the cores $\tilde{S}_3\coloneqq S\cup R_2$, where $B_2\coloneqq \bigcup_{x \in R_2}B_{r(x)}(x)$. Using a covering argument as before, one can easily infer that $\modulus{B_{\lep}(\tilde{S}_3)} \le C_{\lambda, n} \mathcal{F}_{\epsilon}(A)$. Hence, relabeling $\ti_3$ as $A$ and $\tilde{S}_3$ as $S$ (and redefining $R$, $r$ depending on the new field $A$) we are reduced to the case when $R$ consists only points lying outside the $\lep$-neighborhood of $S$ and with $r(x) < \epsilon$. In this case we are not allowed to merely cut off the fields, since we have no control of the singular set in terms of the covering  $V\coloneqq \bigcup_{x \in R} B_{r(x)}(x)$ of $R$ (we can always assume $V$ to be open, i.e. $r(x)>0$ for every $x$). We then need to extend $A$ in a $\Curl$-free way. For, we first notice that using Vitali's Lemma again, we find
\[
 \modulus{V} \le \frac{C}{M}\mathcal{F}_{\epsilon}(A) \le C \epsilon L \alpha \modulus{\log(\alpha)} \le \frac{\lambda}{2} \epsilon.
\]
In particular, this means that every ball of radius $\lambda\epsilon$ must intersect the complement of $V$. We then cover $\Omega_{\lambda\epsilon}(A)$ with (a finite number of) balls of such radius which overlap only finitely many times (depending only on the dimension):
\[
 \Omega_{\lambda\epsilon}(A) \subset \bigcup_{j\ge 1} B_j,\qquad B_j\coloneqq B\rB{x_j, \lambda \epsilon}.
\]
We only need to extend the field to those balls which are not intersecting the singular set (indeed, in those balls which do intersect the singular set we can simply set the field to be a constant). Following the proof of Whitney's Lemma (cf.~\cite{EG}), we define
\[
 \rho(x)\coloneqq \frac{1}{20}\min\cB{1, \dist(x, \mathcal{C})},\qquad \mathcal{C}\coloneqq \Omega_{\lambda\epsilon}(A) \setminus V.
\]
By Vitali's Lemma, we find points $\cB{x_k}\subset V$ such that
\[
 V = \bigcup_{k \ge 1} B\rB{x_k, 5\rho(x_k)},
\]
and the balls $B\rB{x_k, \rho(x_k)}$ are disjoint. One can then prove that the sets
\[
 S_x \coloneqq  \cB{x_k\biggr| B\rB{x, 10\rho(x)}\cap B\rB{x_k, 10\rho(x_k)} \ne \emptyset}, 
\]
have uniformly finite cardinality; more precisely, $\# S_x \le (129)^2=:C_2$ for all $x \in V$. Moreover, if $x_k \in S_x$, $\frac{1}{3}\rho(x_k)\le \rho(x)\le 3\rho(x_k)$. One can then prove that is possible to construct a partition of unity $\cB{\psi_k}_{k \ge 1}$ such that
\[
 \begin{cases}
 	\sum_{k \ge 1} \psi_k(x) \equiv 1,\\
 	\sum_{k \ge 1} \nabla \psi_k(x) \equiv 0,\qquad x\in U,\\
 	\modulus{\nabla \psi_k(x)} \le \frac{C}{\rho(x)}.
 \end{cases}
\]
For each $k$ choose a point $s_k \in \mathcal{C}$ such that $\modulus{x_k - s_k} = \dist(x_k, \mathcal{C})$. Since the balls $B_j$ are simply connected and $A$ is $\Curl$-free in $\Omega_{\lambda\epsilon}(A)$, we can find a function $u \in \mathcal{C}^{\infty}(B_j)$ such that $A = \nabla u_j$ in $B_j$. We can then consider the extension in $B_j$
\[
 \overline{u}_j(x)\coloneqq \begin{cases}
 					 u_j&\text{if }x \in B_j\setminus V,\\
 					 \sum_{k \ge 1}\psi_k(x) \rB{u_j(x) + A(s_k)\rB{x - s_k}}&\text{if } x \in B_j \cap V.
 				    \end{cases}
\]
It is then possible to show that $\overline{A}_j\coloneqq \overline{u}_j \in \mathcal{C}^1(B_j)$ and $\nabla \overline{u}_j(x) = A(x)$ for all $x \in B_j\setminus V$. Moreover, if $x\in B_j \cap B_m$, then $\nabla \overline{u}_j(x) = \nabla \overline{u}_m(x)$. Indeed, since $\nabla u_j = \nabla u_m$ in $B_j \cap B_m$, there exists a constant $c_{jm} \in \mathbb{R}^2$ such that $u_j = c_{jm} + u_m$ in $B_j \cap B_m$, and hence $\nabla \overline{u}_j(x) = \nabla \overline{u}_m(x)$ since $\sum_{k\ge 1}\nabla \psi_k(x) = 0$ for every $x \in V$. In particular, the extension
\[
 \tilde{A}(x)\coloneqq  \overline{A}_j(x),\text{if } x \in B_j
\]
is well defined and $\Curl$-free. It is also easy to verify that $\modulus{\nabla \overline{u}_j(x)} \le C_2$, for some constant $C_2>0$ depending only on the dimension. Then, define $\tilde{A}$ to be the identity in a $2\lambda\epsilon$-neighborhood of $S$, which we call $\tilde{S}$. This gives the desired field $\tilde{A}$, since (arguing like in the discussion before)
\[
 \modulus{B_{\lambda\epsilon}(S)} \le C_2 \modulus{B_{\lambda\epsilon}(\spt\Curl A)},
\]
and
\[
    \int_V \dist^2(\tilde{A}, \so) \de x \le C_2 \sum_j \modulus{B_j} \le \frac{C}{M} \mathcal{F}_{\epsilon}(A, S).\qedhere
\]
 \end{proof}
  \begin{lemma}
  \label{lemma:wlog_bdd_totvar_curl}
  Let $(A, S) \in \mathcal{P}\rB{\epsilon, \alpha, L, \lambda, \ell}$. Then there exists another pair $(\ta, \tilde{S}) \in \mathcal{P}\rB{\epsilon, \alpha, L, \lambda, \frac{\ell}{2}}$ such that for a universal constant $C > 0$
  \begin{enumerate}[(i)]
	  \item $\mathcal{F}_{\epsilon}\rB{\ta, \tilde{S}} \le C \mathcal{F}_{\epsilon}(A, S)$;
	  \item $\Curl(\tilde{A})\in L^{\infty}(\Omega)$ and $\modulus{\Curl(\ta)} \le C \mu_{2, \epsilon}(\tilde{S})$, where
	   \[
		   \mu_{2, \epsilon}(\tilde{S})\coloneqq \frac{1}{\lambda^2 \epsilon} \mathcal{L}^2 \zak B_{\lambda \epsilon}(\tilde{S}).
	   \]
  \end{enumerate}
 \end{lemma}
 \begin{proof}
  By Lemma~\ref{lemma:wlog_bdd}, we can assume $A$ to be smooth in $\Omega_{\lambda \epsilon}(A)$ and $\norm{A}_{L^{\infty}(\Omega)} \le C$. Consider a $\lambda\epsilon$-mollifier $\rho_{\lambda\epsilon}$, that is $\rho_{\lambda\epsilon} \in \mathcal{C}^{\infty}_c\rB{\mathbb{R}^n, [0, 1]}$, $\spt(\rho_{\lambda\epsilon})\subset \overline{B\rB{0, \lambda \epsilon}}$ and $\int\rho_{\lambda\epsilon} = 1$. Take a cut-off function $\zeta$ such that
  \[
   \begin{cases}
    \zeta \in \mathcal{C}^{\infty}(\Omega),\\
    0\le \zeta \le 1,\\
    \zeta \equiv 1&\text{in }B_{\lambda\epsilon}(S),\\
    \zeta \equiv 0&\text{in }\Omega\setminus B_{2\lambda\epsilon}(S),\\
    \norm{\nabla \zeta}_{L^{\infty}(\Omega)} \le \frac{C_0}{ \lambda \epsilon}.
   \end{cases}
  \]
  Define the new matrix field
  \[
   \ta\coloneqq  \rB{1 - \zeta}A + \zeta \rB{A\star \rho_{\lambda\epsilon}}.
  \]
  Clearly, $\norm{\ta}_{L^{\infty}(\Omega)}\le \norm{A}_{L^{\infty}(\Omega)} \le C$ and
  \[
   \Curl(\ta) = \rB{1 - \zeta} \Curl(A) + \rB{A\star \rho_{\lambda \epsilon} - A}\cdot \nabla^{\perp} \zeta + \zeta A\star\cdot \nabla^{\perp} \rho_{\lambda\epsilon},
  \]
  where we used the notation $v\star \cdot w\coloneqq \sum_{i = 1}^n v_i \star w_i$ for $\mathbb{R}^n$-valued functions $v, w$. In particular,
  \begin{itemize}
   \item in $B_{\lambda \epsilon}(S)$, $\zeta\equiv 1$, hence $\Curl(\ta) = \zeta A \star\cdot \nabla^{\perp} \rho_{\lambda\epsilon}$, which in turn implies $\modulus{\Curl(\ta)} \le \frac{C}{\lambda\epsilon}$;
   \item in $B_{2\lambda \epsilon}(S) \setminus B_{\lambda \epsilon}(S)$, $\Curl(A) = 0$ and so $\Curl(\ta) = \rB{A\star\rho_{\lambda\epsilon} - A}\cdot \nabla^{\perp}\zeta + \zeta A\star\cdot \nabla^{\perp}\zeta$. This gives again $\modulus{\Curl(\ta)} \le \frac{C}{\lambda\epsilon}$;
   \item in $\Omega\setminus B_{2\lambda\epsilon}(S)$, $\Curl(A) \equiv 0$ and $\zeta \equiv 0$, hence $\Curl(\ta) = 0$.
  \end{itemize}
  From the discussion above, we have in particular that $\spt\Curl(\ta) \subset B_{3\lambda\epsilon}(S)=:\tilde{S}$. Thus, for every $E\subset\Omega$
  \[
   \begin{split}
   \modulus{\Curl(\ta)}(E) &= \int_E \modulus{\Curl\ta}\de x = \int_{E\cap \spt\Curl(\ta)} \modulus{\Curl\ta}\de x \le \\
                           &\le \norm{\Curl(\ta)}_{L^{\infty}(\Omega)}\modulus{E\cap B_{\lambda \epsilon}(\spt\Curl(\ta))} \\
                           &\le C\mu_{2,\epsilon}(\ta)\sB{E}.
   \end{split}
  \]
  Moreover, a standard covering argument gives
  \[
   \modulus{B_{\lambda \epsilon}(\tilde{S})} \le C \modulus{B_{\lambda \epsilon}(S)},
  \]
  which leads also to
  \[
   \begin{split}
     \int_{\Omega} \dist^2(\ta, \so) \de x &= \int_{\Omega\setminus B_{2\lambda \epsilon}} \dist^2(A, \so) \de x + \int_{B_{2\lambda \epsilon}(S)} \dist^2(\ta, \so)\de x\\
                                           &\le E_{\text{el}}(A, S) + C \modulus{B_{\lambda\epsilon}(S)} \le C\mathcal{F}_{\epsilon}(A, S).\qedhere
   \end{split}
  \]
 \end{proof}

Now we start the analysis of the structure of limits of energy-minimizing tensor fields, proving a compactness result for sequences of vanishing elastic energy and bounded total variation of the $\Curl$. The geometric rigidity estimates in~\cite{MSZ} and~\cite{LL} are crucial in the proof.
 \begin{proposition}
  \label{prop:BV}
  Let $\Omega\subset\mathbb{R}^n$, $n \ge 2$, be a bounded, open, simply connected set, and consider a sequence of matrix fields $A_j \in L^2(\Omega)^{n \times n}$ such that
  \[
    \lim_{j \to \infty} \norm{\dist(A_j, SO(n))}_{L^2(\Omega)} = 0 \qquad\text{ and }\qquad \sup_{j \ge 1} \modulus{\Curl(A_j)}(\Omega) < \infty.
  \]
  Then, up to a subsequence, $\cB{A_j}$ converges strongly in $L^2(\Omega)$ to a matrix field $A \in BV(\Omega, SO(n))$. Moreover, the set of points where $A$ does not belong to $SO(n)$ has Hausdorff dimension at most $n-1$ and, for a dimensional constant $C > 0$,
  \begin{equation}
   \label{eq:curl_bounds_D}
   \modulus{DA}(\Omega) \le C \modulus{\Curl(A)}(\Omega).
  \end{equation}
 \end{proposition}
 \begin{proof}
	 Since $\int_{\Omega}{\modulus{A_j}^2}\de x \le C$, there exists (up to a subsequence) a matrix field $A \in L^2(\Omega)$ such that $A_j \rightharpoonup A$ in $L^2(\Omega)$.
	 Fix $x \in \Omega$ and $\rho > 0$. Using either~\cite[Theorem 3.3]{MSZ} when $n = 2$ or ~\cite[Theorem 4]{LL} when $n \ge 3$, we find a rotation $R^j_{\rho, x}$ such that
	 \[
		 \fint_{B(x, \rho)}{\modulus{A_j - R^j_{\rho, x}}^2 \de y} \le \frac{C(B(0, 1))}{\rho^n} \rB{\norm{\dist(A_j, SO(n))} + T_j^{\frac{n}{n-1}}\rB{B(x, \rho)}},
	 \]
	 where $T_j\coloneqq \modulus{\Curl(A_j)}\xrightharpoonup{*} T$. Thus, taking the $\limsup$ and passing to a subsequence, we find
	 \[
		 \lim_{j \to \infty}{\fint{\modulus{A_j - R^j_{\rho, x}}^2\de y}} \le C\limsup_{j \to \infty}{\frac{T_j(B(x, \rho))^{\frac{n}{n-1}}}{\rho^n}} \le C\frac{T(B(x, 2\rho))^{\frac{n}{n-1}}}{\rho^n}.
	 \]
	 Up to another subsequence, $R^j_{\rho, x} \to R_{\rho, x} \in SO(n)$. So, by the lower semi-continuity of the $L^2$ norm,
	 \begin{equation}
		\label{eq:aaa}
		 \int_{B(x, \rho)}{\modulus{A - R_{\rho, x}}^2\de y} \le  CT(B(x, 2\rho))^{\frac{n}{n-1}}.
	 \end{equation}
   We can now prove that $A(x) \in SO(n)$ for all $x \in \Omega\setminus \mathcal{M}$, where $\mathcal{M}$ has Hausdorff dimension at most $n-1$. For, notice that
   \[
    \mathcal{M}:=\cB{x \in \Omega\biggr| A(x) \notin SO(n)} \subset \bigcup_{c > 0} \mathcal{M}_c,
   \]
	 where
	 \[
	  \mathcal{M}_c:=\cB{x \in \Omega \biggr| \exists \rho_i \downarrow 0\text{ s.t. }\int_{B(x, \rho_i)} \modulus{A - R_{\rho_i, x}}^2 \de y > c\rho_i^n }.
	 \]
	 For each $x \in \mathcal{M}_c$ and $\delta > 0$, we can find a $\rho_{\delta}(x) \le \delta$ such that
	 \begin{equation}
	   \label{eq:565}
  	 CT(B(x, 2\rho_{\delta}(x)))^{\frac{n}{n-1}} \ge \int_{B(x, \rho_{\delta}(x))} \modulus{A - R_{\rho_{\delta}(x), x}}^2 \de y > c\rho_{\delta}(x)^n.
	 \end{equation}
	 Using Vitali's Lemma, we can find countably many points $x_i \in \mathcal{M}_c$ such that the balls $B_i \coloneqq B(x_i, 2\rho_{\delta}(x_i))$ are disjoint and
	 \[
	  \begin{split}
  	  \mathcal{M}_c &\subset \bigcup_{x \in \mathcal{M}_c} B(x, \rho_{\delta}(x)) \subset \bigcup_{x \in \mathcal{M}_c} \ol{B(x, 2\rho_{\delta}(x))} \subset\\
  	                &\subset \bigcup_{x \in \mathcal{M}_c} B(x, 6\rho_{\delta}(x)).
	  \end{split}
	 \]
	 Using~\eqref{eq:565}, we find
	 \begin{equation}
	  \label{eq:preH}
	  \mathcal{H}^{n-1}_{12 \delta}(\mathcal{M}_c) \le C \sum_{i \ge 1} \rho_{\delta}(x_i)^{n-1} \le C T(\Omega).
	 \end{equation}
	 Using the fact that the sets $\mathcal{M}_c$ are decreasing and taking the limit as $\delta\to 0$ in~\eqref{eq:preH}, we obtain
	 \[
	  \mathcal{H}^{n-1}(\mathcal{M}) \le C T(\Omega).
	 \]
	 
	 In particular, $A(x) \in SO(n)$ for almost every $x \in \Omega$. We can then apply~\cite[Proposition 1]{LL}, which gives~\eqref{eq:curl_bounds_D} (and $A \in BV(\Omega, SO(n))$).\\
	 Moreover, $A_j \rightharpoonup A$ and $A(x) \in \so$ for almost every $x \in \Omega$. Denote as $R_j(x)$ the projection of $A_j(x)$ on $\so$. Then $A_j = R_j + (A_j - R_j)$. We know that $A_j - R_j \to 0$ in $L^2(\Omega)$ while, up to a subsequence, $R_j \rightharpoonup A$. But then $R_j \to A$ (because the $L^2$ norms converge to the norm of $A$), and thus $A_j \to A$ in $L^2(\Omega)$.
	 \end{proof}
	From Lemmas~\ref{lemma:wlog_bdd_totvar_curl} and~\ref{lemma:wlog_bdd} and Proposition~\ref{prop:BV},  we obtain immediately the following corollary:
	 \begin{corollary}
	 \label{cor:BV}
	 There exists a constant $C> 0$ such that if $\epsilon_j \to 0$ and $(A_j, S_j) \in \mathcal{P}_{\epsilon_j, \alpha, L}$ be such that $\mathcal{F}_j(A_j, S_j)\le E_{\text{gb}}(\epsilon_j)$. Then, there exists another sequence, still denoted by $(A_j, S_j)$, such that $\mathcal{F}_j(A_j, S_j)\le C E_{\text{gb}}(\epsilon_j)$, $A_j \to A \in BV(\Omega)$ in $L^2(\Omega)$ and $A(x) \in \so$ for every $x \in \Omega \setminus \mathcal{M}$, where $\mathcal{M} \subset \Omega$ is a set of Hausdorff dimension at most $1$.
	 \end{corollary}

Using a slicing argument and Corollary~\ref{cor:BV}, we obtain the estimate $\mu_2(\Omega)\ge \modulus{DA}(\Omega)\ge C \alpha L$, that is a (weak) lower bound to the energy. We are going to improve this result in a first qualitative, and then quantitative way. By qualitative we mean that the limit field is actually a microrotation, while the quantitative improvement is the estimate gives a lower bound comparable to the energy of a grain boundary. These facts rely essentially on two basic tools: the existence of a harmonic competitor and an ``optimal foliation'' lemma. We give here the proof of the first one.
\begin{proposition}
	\label{proposition:wlog_harmonic}
	Let $\Omega \subset \mathbb{R}^n$ be open, and $A\in L^{\infty}(\Omega)^{n\times n}$ be a matrix field such that $\norm{A}_{\infty} \le M$, and let $O \subset \Omega \setminus B_{\lambda \epsilon}(\spt \Curl A)$ be an open, connected subset with Lipschitz boundary. Then there exists a matrix field $\ta \in L^2(\Omega)^{n\times n}$ which is harmonic in $O$ and a constant $C_{n, M }> 0$ (depending only on the dimension $n$ and $M$) such that
	\[
		\norm{A - \ta}_{L^2(O)} \le C_{n, M} \norm{\dist(A, \SOn)}_{L^2(O)}.
	\]
 \end{proposition}
\begin{proof}
	Let $E\coloneqq \norm{\dist(A, \SOn)}_{L^2(O)}^2$. The Hodge decomposition of $A$ gives a vector field $u \in W^{1, 2}_0(\Omega)^n$ and a divergence-free (in the sense of distributions in $\Omega$) matrix field $F \in L^2(\Omega)^{n\times n}$ such that
	\[
		A = \nabla u + F.
	\]
	As in the proof of Lemma~\ref{lemma:wlog_bdd}, we can assume $A$ to be smooth in $\ol{B_{\lambda \epsilon}(\spt \Curl A)}$. Consider the function $u^1_h \in W^{1,2}(O)$ defined as the harmonic extension of $u$ in $O$:
    \[
        \begin{cases}
            \Delta u^1_h = 0 &\text{in } O,\\
            u^1_ h = u^1 &\text{on } \partial O,
        \end{cases}
    \]
    and let then $A_h\coloneqq \nabla (u^1_h, u^2, \cdots, u^n) + F$. Define $G\coloneqq O \cap \cB{\det(A) > \frac{1}{2}}$, and $U(x)\coloneqq \sqrt{A A^T}\chi_G + (1-\chi_G)\id$, together with the vector fields $R_i(x)\coloneqq \sB{U(x)^{-1}\rB{\nabla u + F}}^i$ and $R_{1h}\coloneqq \sB{U(x)^{-1}\rB{\nabla u_h + F}}^1$. In what follows, we identify vector fields with their associated differential $1$-forms.
    We first notice that
    \begin{equation}
	    \label{eq:eq_det}
	    \int_O \det(A)\de x = \int_O \det(A_h) \de x.
    \end{equation}
    Indeed, since the determinant is a null Lagrangian,$\:$~\eqref{eq:eq_det} is equivalent to
    \[
      \sum_{i = 2}^n \int_O \de \rB{u^1_h - u^1}\wedge \bigwedge_{j=2}^n \rB{(1 - \delta_{ij})\de u^j + F^j} = 0,
    \]
    which holds because of the Leibniz formula for forms, the fact that $\Curl F^i = 0$ in $O$ and Stokes' theorem (together with $u^1_h = u^1$ on $\partial O$). Hence, we can write (notice that, since $R_1, \cdots, R_n$ are orthonormal, for any vector field $\mathcal{A}$ we have $\mathcal{A} \wedge R_2\wedge \cdots \wedge R_n = \sum_{k = 1}^n\scal{\mathcal{A}}{R_k}R_k\wedge R_2 \wedge \cdots \wedge R_n = \scal{\mathcal{A}}{R_1}R_1 \wedge \cdots \wedge R_n = \scal{\mathcal{A}}{R_1}\de x^1 \wedge \cdots \wedge \de x^n$)
    \begin{equation}
    	\label{eq:h1}
        \begin{split}
            \int_O \det(A) \de x &= \int_O \det(A_h) \de x = \int_O \rB{\de u^1_h + F^1}\wedge \rB{\de u^2 + F^2} \wedge \cdots \wedge \rB{\de u^n + F^n} =\\
            &=\int_O \rB{R_{1 h} \wedge R_2 \wedge \cdots \wedge R_n} \det(U)  = \\
            &=\int_G \scal{R_{1h}}{R_1} \det(U)\de x + \int_{O \setminus G} R_{1h}\wedge R_2 \wedge \cdots \wedge R_n \det(U).
        \end{split}
    \end{equation}
    On the other hand,
    \begin{equation}
    	\label{eq:h2}
        \int_O \det(A) \de x = \int_G \scal{R_1}{R_1} \det(U) + \int_{O\setminus G} R_1 \wedge \cdots \wedge R_n \det(U).
    \end{equation}
    Subtracting~\eqref{eq:h1} from~\eqref{eq:h2}, we obtain
    \begin{equation}
        \label{eq:eq1}
        0 = \int_G \scal{R_1 - R_{1h}}{R_1}\det(U) \de x + \int_{O \setminus G} \rB{R_1 - R_{1h}}\wedge R_2 \wedge \cdots \wedge R_n\det(U).
    \end{equation}
    Rewrite~\eqref{eq:eq1} as
    \[
        \begin{split}
        \int_G \scal{\nabla u^1 - \nabla u^1_h}{\nabla u^1 + F^1} = &-\int_G \scal{\nabla u^1 - \nabla u^1_h}{\rB{\det(U)U^{-2} - \id}\rB{\nabla u^1 + F^1}} \de x +\\
        &+\int_{O\setminus G} \rB{R_1 - R_{1h}} \wedge R_2 \wedge \cdots \wedge R_n \det(U),
    \end{split}
    \]
    and then add $\int_{O\setminus G} \scal{\nabla u^1 - \nabla u^1_h}{\nabla u^1 + F^1}$ on both sides. Since $u^1_h$ is the harmonic extension of $u^1$ and $\text{div}(F^1) = 0$, we have
    \[
	\int_O \modulus{\nabla u^1 - \nabla u^1_h}^2 \de x = \int_O \scal{\nabla u^1 - \nabla u^1_h}{\nabla u^1 + F^1}\de x = I_1 + I_2 + I_3,
    \]
    where
    \[
        \begin{split}
            I_1&\coloneqq \int_{O\setminus G} \scal{\nabla u^1 - \nabla u^1_h}{\nabla u^1 + F^1},\\
            I_2&\coloneqq -\int_G \scal{\nabla u^1 - \nabla u^1_h}{\rB{\det(U)U^{-2} - \id}\rB{\nabla u^1 + F^1}},\\
            I_3&\coloneqq \int_{O\setminus G} \rB{R_1 - R_{1h}}\wedge R_2 \wedge \cdots \wedge R_n.
        \end{split}
    \]
    Now, because of the continuity of the determinant, there exists a dimensional constant $c_n > 0$ such that $\cB{\det(A) \le \frac{1}{2}} \subset \cB{\dist(A, \SOn)\ge C_n}$. Thus, since $\norm{A}_{\infty} \le M$,
    \[
        \begin{split}
            \modulus{\int_{O\setminus G} \scal{\nabla u^1 - \nabla u^1_h}{\nabla u^1 + F^1}\de x} &\le M \int_{\cB{\dist(A, SO(n)) \ge C_n}} \modulus{\nabla u^1 - \nabla u^1_h} \de x \le \\
            &\le C_n M \sqrt{E} \sqrt{\int_O \modulus{\nabla u^1 - \nabla u^1_h}^2 \de x}.
        \end{split}
    \]
    Let us now estimate $I_2$. Since the function $f(U)\coloneqq U^{-2}\det(U)$ is smooth on $G$, $\norm{A}_{\infty}\le M$ and $f(\id) = \id$, there exists a constant, depending only on $n$ and $M$,$C_n = C_n(M) > 0$ such that
    \[
        \modulus{f(U) - \id} \le C_n \modulus{U - \id} = C_n\dist(A, \SOn).
    \]
    Then
    \[
        \begin{split}
        \modulus{I_2} &\le C_n \int_G \modulus{\nabla u^1 - \nabla u^1_h} \modulus{U^{-2}\det(U) - \id}\modulus{\nabla u^1 + F^1} \de x \le\\
        &\le C_n \sqrt{E} \sqrt{\int_O \modulus{\nabla u^1 - \nabla u^1_h}^2 \de x}.
    \end{split}
    \]
    Finally, let us estimate $I_3$. Again because of the boundedness of $A$,
    \[
        \begin{split}   
            \modulus{I_3} \le C_n \int_{O\setminus G} \modulus{R_1 - R_{1h}}\frac{c_n}{c_n}\de x \le C_n \sqrt{E}\sqrt{\int_O \modulus{\nabla u^1 - \nabla u^1_h}^2 \de x}.
        \end{split}
    \]
    Combining these estimates together, we find
    \[
        \int_O \modulus{\nabla u^1 - \nabla u^1_h}^2 \de x \le C_n \sqrt{E}\sqrt{\int_O \modulus{\nabla u^1 - \nabla u^1_h}^2 \de x},
    \]
    i.e.
    \[
        \int_O \modulus{\nabla u^1 - \nabla u^1_h}^2 \de x \le C_n E.
    \]
    Applying the same procedure to each component, we find
    \[
        \int_O \modulus{\nabla u - \nabla u_h}^2 \de x \le C_n E,
    \]
    where $u_h = (u^1_h, \cdots, u^n_h)$. Now we can define
    \[
	    \widetilde{u}\coloneqq u_h \chi_O + u\chi_{\Omega\setminus O},
    \]
    and set $\ta\coloneqq \nabla\widetilde{u} + F$. Since $\Div(\ta) = \Delta \widetilde{u} = 0$ and $\Curl(\ta) = 0$ in $O$, from the identity
    \[
	    -\Delta L^i_j + \partial_j \Div L^i = -\sum_{k = 1}^n\partial_k\rB{\partial_k L^i_j - \partial_j L^i_k},
    \]
    valid for any matrix field $L \in L^1\rB{\Omega}^{n\times n}$, we infer that $\Delta\ta = 0$ in $O$.
\end{proof}

\begin{remark}
 \label{rmk:competitors}
 Combining together the lemmata~\ref{lemma:wlog_bdd},~\ref{lemma:wlog_bdd_totvar_curl} and Proposition~\ref{proposition:wlog_harmonic}, we have that for every $(A, S)\in \mathcal{P}_{\epsilon}$ such that $\mathcal{F}_{\epsilon}(A, S)\le E_{\text{gb}}(\epsilon)$, we can find a competitor $(\tilde{A}, \tilde{S})\in \mathcal{A}_{\epsilon}$ whose energy can be estimated in terms of the original one, i.e. $\mathcal{F}_{\epsilon}(\tilde{A}, \tilde{S})\le C \mathcal{F}_{\epsilon}(A, S)$, where $C>0$ is a universal constant, satisfying the following properties:
 \begin{enumerate}[(a)]
 	\item $\norm{\tilde{A}}_{\infty} \le C$;
 	\item $\modulus{\Curl(\tilde{A})} \le C\mu_{2, A}$;
 	\item $\Delta \tilde{A} = 0$ in $\Omega_{\lambda\epsilon}(\tilde{A})$.
 \end{enumerate}
 That is, since we are interested in a lower bound to the energy, we can restrict our attention to those pairs in $\mathcal{P}_{\epsilon}$ satisfying (a), (b) and (c).
\end{remark}

\begin{remark}
	If $A \in \mathcal{A}_{\epsilon} \cap \cB{G:\norm{G}_{\infty} \le M}$ and $\ta$ is the matrix field given by Lemma~\ref{proposition:wlog_harmonic}, then the Burgers' vectors relative to $A$ still define a bounded functional from $1$-cycles into $\mathbb{R}^2$, and it can also be proved without employing the maximum principle. Indeed, if we identify $\tilde{A}$ with a vector of $1$-forms, the Burgers' vector
	\[
		\declareapp{\harpoon{b}_{\ta}}{Z_1(\Omega_{\lambda\epsilon}(\ta); \mathbb{R})}{\mathbb{R}^2}{T}{\scal{T}{\ta}}
	\]
	defines a bounded operator (where the space of $1$-cycles is endowed with the mass norm). Indeed, in $\Omega\setminus B_{\lambda \epsilon}(\ta)$, we can write $\ta = \de u_h + F$, where $A = \de u + F$. Then, since $T$ is a closed current,
	\[
		\scal{T}{\ta} = \scal{T}{\de u_h + F} = \scal{T}{F} = \scal{T}{\de u + F} = \scal{T}{A}.
	\]
	But $\modulus{\scal{T}{A}} \le \norm{A}_{\infty} \mathbf{M}(T)\le M \mathbf{M}(T)$, hence the claim.
\end{remark}
We shall need the following Lemma, which gives an expression for the Burgers' vector in terms of the gradient of the fields and the position of the points on the curve.
\begin{lemma}
    \label{lemma:repr_burgers}
		Suppose $\gamma\subset \mathbb{R}^2$ is a closed, simple Lipschitz curve, and $V$ is a $\mathcal{C}^1$ vector field defined in a neighborhood of $\gamma$. Then
		\[
			\int_{\gamma} V(x) \cdot t(x) \de \mathcal{H}^1 = -\int_{\gamma} \nabla V(x) x \cdot t(x) \de \mathcal{H}^1,
		\]
		where $t(x)$ is the tangent vector of $\gamma$ at $x$.
	\end{lemma}
	\begin{proof}
		Let $\gamma = \cB{f(t)\biggr|t\in[0, 1)}$, where $f$ is a Lipschitz parametrization of $\gamma$, and set $x_0\coloneqq f(0)=f(1)$. Then
		\[
			\begin{split}
				\int_{\gamma} V \cdot t \de \mathcal{H}^1 &= \int_{\gamma}\rB{V(x) - V(x_0)}\cdot t\de \mathcal{H}^1 = \int_0^1 \rB{V(f(t)) - V(f(0))}\cdot \dot{f}(t)\de t = \\
				&=\int_0^1\rB{\int_0^t \nabla V(f(s))\dot{f}(s)\de s}\cdot \dot{f}(t)\de t = \int_0^1 \nabla V(f(s))\dot{f}(s)\cdot \int_s^1 \dot{f}(t) \de t \de s = \\
				&= -\int_0^1\nabla V(f(s))\dot{f}(s) \cdot f(s) \de s = - \int_0^1 \nabla V(f(s))f(s)\cdot \frac{\dot{f}(s)}{\modulus{\dot{f}(s)}}\modulus{\dot{f}(s)} \de s = \\
				&= - \int_{\gamma} \nabla V(x)x\cdot t(x) \de \mathcal{H}^1.\qedhere
		 	\end{split}
		\]
	\end{proof}
	
	As an immediate application of Lemma~\ref{lemma:repr_burgers}, we see that if $\gamma$ lies in a region where $A$ is both $\Curl$ and divergence free, then
	\begin{equation}
			\label{eq:rep_burgers_harmonic}
			\begin{split}
			\harpoon{b}(\gamma) &= -\rB{\int_{\gamma} \rB{x\cdot \nabla A^1_1, x^{\perp}\cdot \nabla A^1_1}\cdot t(x) \de \mathcal{H}^1, \int_{\gamma} \rB{-x^{\perp}\cdot\nabla A^2_2, x\cdot \nabla A^2_2}\cdot t(x)\de \mathcal{H}^1} \equiv \\
			&\equiv - \int_{\gamma} \rB{\rB{\begin{matrix}x & x^{\perp}\\-x^{\perp}& x\end{matrix}}\cdot \rowvect{\nabla A^1_1}{\nabla A^2_2}} t(x)\de \mathcal{H}^1.
			\end{split}
	\end{equation}
We are left with the second fundamental tool, that is the foliation Lemma. In the proof, we will need the following technical covering lemma:
 \begin{lemma}
  \label{lemma:make_deg2_disjoint}
  Let $R > 0$, $\delta\in (0, 1)$, $M > 10$ and consider a family $I = \cB{x_i}_{i = 1}^N$ of points in $\mathbb{R}^n$ whose subfamily $J \subset I$ has the property that for each $j \in J$ there exists a $k\in \mathbb{N}$ such that $\rB{B\rB{x_i, \rB{\frac{M}{2} - 2}\rho_k \setminus B\rB{x_i, \rho_k}}} \cap I = \emptyset$, where $\rho_k\coloneqq \delta^k R$. Set
  \[
	  r(x_j)\coloneqq :r_j\coloneqq \max\cB{\rho_k \biggr| k\ge 0 \text{ and } \rB{B\rB{x_j, \rB{\frac{M}{2} - 2}\rho_k} \setminus B\rB{x_j, \rho_k}} \cap I = \emptyset}.
  \]
  Then there exists a subfamily $\tilde{J}\subset J$ such that the balls $\cB{B\rB{x_i, \rB{\frac{M}{4}-1}r_i}}_{i \in \tilde{J}}$ are disjoint and
  \[
   \bigcup_{j \in J} B\rB{x_j, \rB{\frac{M}{8} - \frac{5}{4}}r_j} \subset \bigcup_{i \in \tilde{J}} B\rB{x_i, \rB{\frac{M}{8} - \frac{1}{4}}r_i}.
  \]
 \end{lemma}
 \begin{proof}
  Let $\beta\coloneqq \frac{M}{2} - 2$. Define inductively the family $\tilde{J}$ as follows. Select a maximal family of points $J_0$ from $\cB{j \in J\biggr| r_j = R}$ such that
  \[
    \modulus{x_i - x_j} \ge \frac{\beta}{2} \rB{r_i + r_j} = \beta r_i = \beta r_j\qquad \forall x_i, x_j \in J_0,
  \]
  and set $\tilde{J}_0\coloneqq J_0$. Suppose then that the family $\tilde{J}_k$ has been defined, $k\ge 0$, and select a maximal family of points $J_{k+1}$ from $\cB{j \in J\biggr| r_j = \delta^{k+1} R}$ such that
  \[
    \modulus{x_i - x_j} \ge \frac{\beta}{2}\rB{r_i + r_j}\qquad \forall x_i, x_j \in \tilde{J}_k \cup J_{k+1},
  \]
  and then set $\tilde{J}_{k+1} = \tilde{J}_k \cup J_{k+1}$. The set $\tilde{J}$ is given by
  \[
    \tilde{J}\coloneqq \bigcup_{k\ge 0} \tilde{J}_k.
  \]
  Clearly, the balls $\cB{B\rB{x_i, \frac{\beta}{2}r_i}}_{i \in \tilde{J}}$ are disjoint. Moreover, for every $x_j \in J$ we can find an $x_i \in \tilde{J}$ such that $r_i = r_j$ and
  \[
    \modulus{x_i - x_j} < \frac{\beta}{2}\rB{r_i + r_j} \le \beta r_i,
  \]
  which means, by the definition of $r_i$, that $\modulus{x_i - x_j} \le r_i$. Hence, if $x \in B\rB{x_j, \rB{\frac{M}{8} - \frac{5}{4}}r_j}$, $j \in J$, then there exists a point $x_i\in\tilde{J}$ such that 
  \[
    \modulus{x - x_i} \le r_i + \rB{\frac{M}{8}-\frac{5}{4}}r_i = \rB{\frac{M}{8}-\frac{1}{4}}r_i.\qedhere
  \]
 \end{proof}
 We are now in position to prove a key step, that is the lemma which gives the optimal foliation.
  \begin{lemma}
    \label{lemma:foliation}
    There exist $\delta_0 \in (0, 1)$ and $C > 0$ such that if $\cB{B(x_i, \rho_i)}_{i=1}^N $ are balls in $\mathbb{R}^2$ satisfying
    \begin{equation}
      \label{eq:small_radii}
      \mathcal{H}^1\rB{\mcA \cap \partial \bigcup_{i=1}^N B(x_i, \rho_i)} \le \delta_0,\quad \mcA\coloneqq B(0, 1) \setminus B\rB{0, \frac{1}{2}}\subset \mathbb{R}^2,
    \end{equation}
    then there exists a Lipschitz function $\phi:\mcA\to[0, 1]$ such that
    \begin{enumerate}[(i)]
      \item $\norm{\nabla \phi}_{L^{\infty}(\mcA)} \le C$;
      \item $\phi \equiv 0$ on $\partial B(0, 1)$ and $\phi \equiv 1$ on $\partial B\rB{0, \frac{1}{2}}$;
      \item If $U\coloneqq \mcA\setminus\uB$,
      \begin{equation}
        \label{eq:fol_energy}
        \int_U \frac{\modulus{\nabla \phi(x)}^2}{\dist^2(x, \partial U)}\de x \le C\rB{1 + N}.
      \end{equation}
    \end{enumerate}
    \end{lemma}
    \begin{proof}
     We shall modify in an appropriate way the natural radial foliation. First of all, define
     \[
      \delta_1\coloneqq \inf\cB{r\ge \delta_0\biggr| \partial B\rB{0, \frac{1}{2}+r} \cap \bigcup_{i = 1}^N B\rB{x_i, \rho_i} = \emptyset},
     \]
     and
     \[
      \delta_2\coloneqq \inf\cB{r\ge \delta_0\biggr| \partial B\rB{0, 1 - r} \cap \bigcup_{i = 1}^N B\rB{x_i, \rho_i} = \emptyset}.
     \]
     By a simple geometric argument, one can see that $\delta_0 \le \min\cB{\delta_1, \delta_2} \le \max\cB{\delta_1, \delta_2} \le \frac{3}{2}\delta_0$. Define then the function
     \[
     \phi_0(x)\coloneqq 
      \begin{cases}
        C(\delta_1, \delta_2) \rB{1 - \delta_2 - \modulus{x}}&\text{if }\modulus{x} \in B\rB{0, 1-\delta_2}\setminus B\rB{0, \frac{1}{2}+\delta_1},\\
        0&\text{if }\modulus{x}\ge 1 - \delta_2,\\
        1&\text{if }\modulus{x}\le \frac{1}{2}+ \delta_1,
      \end{cases}
     \]
     where $C\rB{\delta_1, \delta_2}\coloneqq \frac{1}{\frac{1}{2} - \delta_1 - \delta_2}$ (clearly $\phi_0$ is Lipschitz, with Lipschitz constant $C(\delta_1, \delta_2) \le \frac{1}{\frac{1}{2} - 3\delta_0}$ and satisfies (ii)). We will then split the integral $I$ in the left hand side of~\eqref{eq:fol_energy} in three terms: one where, roughly speaking, we see enough space in order to interpolate the function with a constant, another one where the balls accumulate (where we will use a covering argument) and a last one where we are very close to the balls $B\rB{x_i, \rho_i}$ (of which we will get rid of simply by using a ``cutting-out'' function, possible because of~\eqref{eq:small_radii}).\\
     In order to detect the regions where we have to modify the foliation, it is convenient to introduce particular coverings and organize them in a graph. For, define the sets
     \[
      U_k\coloneqq \cB{x \in U\biggr| r_{k-1} < \dist(x, \points) \le r_k},
     \]
where $r_k\coloneqq M^k r_0$ and $r_0\coloneqq \frac{c_0}{n}$, for some constants $c_0 > 0$ and $M>2$ to be chosen later, and $k \in \cB{1, \cdots, K}$, $K\coloneqq \sB{\frac{1}{2}\log(N)}$.
Let $I\coloneqq \cB{x_1, \cdots, x_N}$ and for each $k \in \cB{0, \cdots, K}$ choose a maximal family $I_k$ of points in $I$ whose reciprocal distances are $\ge r_k$. Notice that for each $k$ the balls $\cB{B\rB{x_i, 2r_k}}_{i \in I_k}$ are a cover of $U_k$. We then define the \emph{edge maps}
\[
 E_k: I_k \longrightarrow I_{k+1},
\]
which have the property that, for $x_i \in I_k$,
\[
 \modulus{E_k(x_i) - x_i} = \min\cB{\modulus{x_j - x_i}\biggr| x_j \in I_{k+1}}.
\]
Clearly, $\modulus{E_k(x_i) - x_i} < r_{k+1}$; indeed, either $x_i \in I_{k+1}$ (and in such a case $E_k(x_i) = x_i$) or $x_i \notin I_{k+1}$. But then $\modulus{x_j - x_i} < r_{k+1}$ for some $j\in I_{k+1}$ in order to not contradict the maximality of $I_{k+1}$.\\
We can now define the directed graph (actually, the forest) $G=(V, E)$ whose vertices are given by
\[
 V\coloneqq \cB{(x_i, k)\biggr| x_i \in I_k,\quad k \in \cB{1, \cdots, K}},
\]
and whose edges are
\[
 E\coloneqq \cB{\rB{(x_i, k), (E_k(x_i), k+1)}\biggr| i \in I_k,\quad k \in \cB{1, \cdots, K-1}}.
\]
We write $v\sim w$ if either $(v, w)\in E$ or $(w, v)\in E$.
Notice that $G$ is the disjoint union of (directed) trees whose roots are the points $(x_i, K)$, $x_i \in I_K$. Given a vertex $v = (x_i, k)$, we denote by $T_v$ the subtree rooted at $v$. We also define the  ``pruned'' tree at the vertex $v$ as
\[
 T^{\text{pr}}_v\coloneqq T_v \setminus \bigcup_{\substack{v'\in V\\ T_{v'}\subset T_v, \quad T_{v'}\ne T_v\\\DEG(v')=2}} T_{v'}\setminus\cB{v'}.
\]
We then have the pruned forest
\[
 G^{\text{pr}}\coloneqq \bigcup_{x_i \in I_K} T^{\text{pr}}_{(x_i, K)} =:\rB{V^{\text{pr}}, E^{\text{pr}}}.
\]
The vertices of degree $2$ where we prune the tree are the ones which we will see to correspond to empty annuli. To see this, notice that if a vertex $v = (x_i, k) \in V$, $k \le K-1$, has degree $2$ and $v' = (x_j, k') \in T_v$ with $k'\le k-1$, then
\begin{enumerate}[(a)]
	\item $I_{k-1} \cap B\rB{x_i, \frac{r_k}{2}} = \cB{x_{i_0}(x_i)} \equiv \cB{x_{i_0}}$ is a singleton. Indeed, 
 \[
  2 = \DEG(x_i) = \#E_{k-1}^{-1}(x_i) + 1.
 \]
 But $E_{k-1}^{-1}(x_i) \supset I_{k-1}\cap B\rB{x_i, \frac{r_k}{2}}$ (indeed, if $x_l \in I_{k-1} \cap B\rB{x_i, \frac{r_k}{2}} \cap I_{k-1}$, then for every $x_{i'} \in I_k\setminus \cB{x_i }$ we have $\modulus{x_l - x_{i'}} \ge \modulus{x_i - x_{i'}} - \modulus{x_i - x_j} > r_k - \frac{r_k}{2} = \frac{r_k}{2} > \modulus{x_l - x_i}$, that is $E_{k-1}(x_l) = x_i$) which is always not empty. Otherwise, $x_i \notin I_{k-1}$ and for every $x_j \in I_{k-1}$ we have $\modulus{x_j - x_i} > \frac{r_k}{2} > r_{k-1}$, i.e. $\cB{x_i} \cup I_{k-1}$ would be a family whose points have reciprocal distance is $\ge r_{k-1}$ and which strictly contains $I_{k-1}$, which was assumed to be a maximal family. Hence,
 \[
 1 = \# E_{k-1}^{-1}(x_i) \ge \# I_{k-1}\cap B\rB{x_i, \frac{r_k}{2}} \ge 1,
 \]
 i.e. $\# I_{k-1}\cap B\rB{x_i, \frac{r_k}{2}} = 1$, say $I_{k-1}\cap B\rB{x_i, \frac{r_k}{2}} = \cB{x_{i_0}}$;
 \item $\modulus{x_i - x_{i_0}} < r_{k-1}$. This is clear, because of what we said at the point (a);
 \item $\rB{B\rB{x_i, \frac{r_k}{2} - r_{k-1}}\setminus B\rB{x_{i_0}, r_{k-1}}} \cap I = \emptyset$. This is also a direct consequence of the previous two points. In particular,
	 \[
		 \rB{B\rB{x_{i_0}, \frac{r_k}{2} - 2r_{k-1}}\setminus B\rB{x_{i_0}, r_{k-1}}} \cap I = \emptyset;
	 \]
 \item $\modulus{x_j - x_{i_0}} < \frac{M}{M - 1}r_{k-1} = \frac{r_k}{M - 1}$. Indeed, since $e$ has degree $2$, the only vertex at level $k-1$ is precisely $\rB{x_{i_0}, k - 1}$. Hence, if we set $y_0\coloneqq x_j$ and define inductively $y_{i+1}\coloneqq E_{k'+i}(y_i)$, $i = 0, \cdots, k - k' - 2$, we have (since $E_{k-2}(y_{k - k' - 2}) = x_{i_0}$)
 \[
  \modulus{x_j - x_{i_0}} \le \sum_{i = 0}^{k - k' - 2} \modulus{y_{i+1} - y_i} \le r_{k-1} \sum_{i = 0}^{k - k' -2} M^{-i} \le r_{k-1} \frac{M}{M-1} < \frac{r_k}{M - 1}.
 \]
 
\end{enumerate}

Define then the family of points $J$ as
 \[
	 J\coloneqq \cB{x_{i_0}(x_i)\biggr| x_i \in I \text{ and }\DEG((x_i, k)) = 2\quad \text{ for some }k \in \cB{2, \cdots, K-1}} \subset I.
 \]
 Lemma~\ref{lemma:make_deg2_disjoint} gives a subfamily $\tilde{J}$ such that
 \begin{itemize}
	 \item $\cB{B\rB{x_{i_0}(x_i), \rB{\frac{M}{4} - 1}\ol{r}_i}}_{i \in \tilde{J}}$ are disjoint;
	 \item $\sB{B\rB{x_{i_0}(x_i), \rB{\frac{M}{2} - 2}\ol{r}_i}\setminus B\rB{x_{i_0}(x_i), \ol{r}_i}} \cap I = \emptyset$;
  \item Provided $\frac{M}{8} - \frac{5}{4} > 3$,
  \[
	  \bigcup_{j \in J} B\rB{x_{i_0}(x_j), 3\ol{r}_j} \subset \bigcup_{j \in \tilde{J}} B\rB{x_j, \rB{\frac{M}{8} - \frac{1}{4}}\ol{r}_j},
  \]
   \end{itemize}
   where
   \[
	   \ol{r}_i\coloneqq \max\cB{r_k\biggr|B\rB{x_{i_0}(x_i), \rB{\frac{M}{2} - 2}r_k}\setminus B\rB{x_{i_0}(x_i), r_k} \cap I = \emptyset }.
   \]
  Let now $c_1\coloneqq \frac{M}{8}+\frac{3}{4}$ and $c_2\coloneqq \frac{M}{4}-2$, and consider the Lipschitz function $\eta:\mathbb{R}^+\to [0, 1]$
  \[
   \eta(t)\coloneqq \begin{cases}
              1&\text{if } t\in[0, c_1],\\
              \frac{1}{c_1 - c_2}t - \frac{c_2}{c_1 - c_2}&\text{if }t\in[c_1, c_2],\\
              0&\text{if }t\ge c_2,
            \end{cases}
  \]
  whose Lipschitz constant is $\frac{1}{c_1 - c_2}$. Define
  \[
	  \phi_1(x)\coloneqq \sum_{j\in\tilde{J}} \rB{ \eta\rB{\frac{\modulus{x - x_{i_0}(x_j)}}{\ol{r}_j}}\ol{\phi}_{0, j} + \rB{1 - \eta\rB{\frac{\modulus{x - x_{i_0}(x_j)}}{\ol{r}_j}}}\phi_0(x)},
  \]
  where $\ol{\phi}_{0, j} = \fint_{B\rB{x_{i_0}(x_j), c_2 \ol{r}_j}} \phi_0(y)\de y $. Since the balls defined by the family $\tilde{J}$ are disjoint, we easily infer
  \[
   \begin{split}
	   \norm{\nabla \phi_1}_{\infty} &\le \norm{\nabla \phi_0}_{\infty} + \max_{j \in \tilde{J}} \rB{\modulus{\eta'}_{\infty} \frac{1}{\ol{r}_j} \norm{\phi_0 - \ol{\phi}_{0, j}}_{L^{\infty}\rB{B(x_{i_0}(x_j), c_2 \ol{r}_j)}}} \le\\
     &\le C(\delta_0, M).
   \end{split}
  \]
  Finally, consider the set $\mathcal{I}\coloneqq \phi_1\rB{\bigcup_{i=1}^N B\rB{x_i, 2\rho_i + r_0}}$. Then
  \[
   \mathcal{L}^1\rB{\mathcal{I}} \le \Lip(\phi_1)\sum_{i = 1}^N\rB{4\rho_i + 2r_0} \le \frac{1}{2},
  \]
  provided we take $\delta_0 \le \frac{1}{16\Lip(\phi_1)}$ and $c_0=2\delta_0$. Consider then the Lipschitz function $\psi:[0, 1]\to [0, 1]$ defined by
  \[
   \psi'\coloneqq \frac{\chi_{[0, 1] \setminus \mathcal{I}}}{1 - \modulus{\mathcal{I}}},\qquad \psi(0) = 0,\qquad \psi(1) = 1.
  \]
  Define $\phi\coloneqq \psi\circ\phi_1$. Clearly $\phi$ satisfies (i) and (ii). Let us prove it satisfies (iii). For, notice first that if $x \notin \bigcup_{i=1}^N B\rB{x_i, 2\rho_i}$, then $d(x)\coloneqq \dist\rB{x, \cB{x_i}_{i=1}^N} \le 2\dist\rB{x, \partial U}$. Set $U'\coloneqq \rB{B(0, 1)\setminus B\rB{0, \frac{1}{2}}} \setminus \bigcup_{i=1}^NB\rB{x_i, 2\rho_i + r_0}$ and $U''\coloneqq \bigcup_{i=1}^NB\rB{x_i, 2\rho_i + r_0}\setminus \bigcup_{i=1}^NB\rB{x_i, \rho_i}$. Since $\phi$ is constant on $\bigcup_{i=1}^NB\rB{x_i, 2\rho_i + r_0}$
  \[
   \begin{split}
    \int_U \frac{\modulus{\nabla \phi}^2}{\dist^2(x, \partial U)}\de x&\le C\rB{ \int_{U'} \frac{\modulus{\nabla \phi}^2}{d^2(x)} \de x + \int_{U''} \frac{\modulus{\nabla \phi }^2}{\dist^2(x, \partial U)}\de x} = \\
    &= C \int_{U'} \frac{\modulus{\nabla \phi}^2}{d^2(x)} \de x \le C \int_{U'\cap\cB{d<\frac{1}{\sqrt{N}}}} \frac{\modulus{\nabla \phi}^2}{d^2(x)} \de x + N.
   \end{split}
  \]
  Write
  \[
   U'\cap\cB{d<\frac{1}{\sqrt{N}}} = U_1' \cup U_2',
  \]
  \[
	  U_1'\coloneqq \rB{U'\cap\cB{d<\frac{1}{\sqrt{N}}}}\setminus \bigcup_{j \in \tilde{J}}B\rB{x_{i_0}(x_j), \rB{\frac{M}{4} - 2} \ol{r}_j},
  \]
  \[
  U_2' = U \setminus U_1'.
  \]
  Notice that $U_k \cap U_1' \subset \bigcup_{i\text{ : }(x_i, k)\in V^{\text{pr}}} B\rB{x_i, 2 r_k}$. Since any non-trivial tree $T = (V, E)$ satisfies
  \[
  \# V \le 2 \#\cB{v \in V\biggr| \DEG(v) = 1} + \# \cB{v \in V \biggr| \DEG(v) = 2}
  \]
   and the total number of leaves in the forest is always $\le N$, we have
  \[
   \begin{split}
    \int_{U_1'} \frac{\modulus{\nabla \phi}^2}{d^2(x)}\de x &\le \sum_{k = 0}^K \int_{U_k \cap U_1'} \frac{\modulus{\nabla \phi}^2}{d^2(x)}\de x \le C(\delta_0, M) \sum_{k=0}^K \sum_{i\text{ : }(x_i, k)\in V^{\text{pr}}} \frac{1}{r_{k-1}^2}\int_{B\rB{x_i, r_k}}\de x \le \\
    &\le C\rB{\delta_0, M} \# V^{\text{pr}} \le C\rB{\delta_0, M} N.
   \end{split}
  \]
  On the other hand
  \[
   \begin{split}
	   \int_{U_2'} \frac{\modulus{\nabla \phi}^2}{d^2(x)}\de x &\le \sum_{j \in \tilde{J}} \int_{B\rB{x_{i_0}(x_j), \rB{\frac{M}{4} - 2}\ol{r}_j}} \frac{\modulus{\nabla \phi}^2}{d^2(x)}\de x \le \\
                                                             &\le C\rB{\delta_0, M} \# \tilde{J} \le C\rB{\delta_0, M} N.\qedhere
   \end{split}
  \]
    \end{proof}
 By a scaling argument we get the following
 \begin{corollary}
    \label{lemma:foliationR}
    Let $\mcA\coloneqq B(p, 2R)\setminus B\rB{p, R}\subset\mathbb{R}^2$. There exist $\delta_0\in (0, 1)$ and $C > 0$ such that if $\cB{B(x_i, \rho_i)}_{i=1}^N$ are balls satisfying
    \begin{equation}
      \label{eq:small_radiiR}
      \mathcal{H}^1\rB{\mcA \cap \partial \bigcup_{i=1}^N B(x_i, \rho_i)} \le \delta_0 R,
    \end{equation}
    then there exists a Lipschitz function $\phi:\mcA\to[0, 1]$ such that
    \begin{enumerate}[(i)]
      \item $\norm{\nabla \phi}_{L^{\infty}(\mcA)} \le \frac{C}{2R}$;
      \item $\phi \equiv 0$ on $\partial B(p, 2R)$ and $\phi \equiv 1$ on $\partial B\rB{p, R}$;
      \item If $U\coloneqq \mcA\setminus\uB$,
      \begin{equation}
        \label{eq:fol_energyR}
        \int_U \frac{\modulus{x}^2\modulus{\nabla \phi(x)}^2}{\dist^2(x, \partial U)}\de x \le C\rB{1 + N}.
      \end{equation}
    \end{enumerate}
 \end{corollary}
 
 \begin{remark}
  \label{rmk:structure_foliation}
  The proof shows that the foliation $\phi$ constructed in Lemma~\ref{lemma:foliation} is constant on (a neighborhood of) each ball (and in a neighborhood of the boundary of the annulus). Moreover, due to the choice of $\delta_1$ and $\delta_2$, the superlevel sets $\cB{\phi \ge 1} = \cB{\phi = 1}$ and $\cB{\phi > 0}$ contain all the balls $B\rB{x_i, \rho_i}$ they intersect.
 \end{remark}
 
 \begin{remark}
  \label{rmk:howto_foliate}
  We discuss here how to use Corollary~\ref{lemma:foliationR}. Consider a ball $B\rB{q, r}$ and balls $\cB{B\rB{q_i, r_i}}_{i=1}^N$ which satisfy the conditions of Corollary~\ref{lemma:foliationR}, and $(A, S) \in \mathcal{P}_{\epsilon}$ satisfying the conditions in Remark~\ref{rmk:competitors}. Notice that since $\phi$ is constant on each ball $B\rB{q_i, r_i}$, we have $\phi\rB{\bigcup_{i=1}^N B\rB{q_i, r_i}} = \cB{\phi_i}_{i = 1}^{L-1}$. Define $\phi_0\coloneqq 0$ and $\phi_L\coloneqq  1$, and re-label, if necessary, the $\phi_i$ in such a way that $0 = \phi_0 \le \phi_1 < \phi_2 < \cdots < \phi_{L-1} \le \phi_L = 1$. Using the fact that each connected component of $\partial \cB{\phi > h}$ is a closed, simple Lipschitz curve, and that clearly $\cB{\phi_i < \phi < \phi_{i+1}}\cap B_{\lambda \epsilon}(S_{\epsilon_m}) = \emptyset$, we have that for each $h \in \rB{\phi_i, \phi_{i+1}}$
  \[
   \begin{split}
   \int_{\partial \cB{\phi > h}} A_m\cdot t \de \mathcal{H}^1 &= \int_{\cB{\phi > h}} \Curl(A_m) \de x = \int_{\cB{h < \phi < \phi_{i+1}}} \Curl(A_m) \de x + \int_{\cB{\phi \ge \phi_{i+1}}} \Curl(A_m) \de x =\\
   &= \int_{\cB{\phi \ge \phi_{i+1}}} \Curl(A_m) \de x.
   \end{split}
  \]
  Thus, setting $b_i\coloneqq \int_{\cB{\phi = \phi_i}} \Curl(A_m) \de x$, we have
  \[
   \int_{\cB{\phi \ge \phi_i}} \Curl(A_m) \de x = \sum_{j = i}^L b_i.
  \]
  Integrate then for $h \in (0, 1)$ in order to get
  \[
   \begin{split}
    \int_0^1 \de h \int_{\partial \cB{\phi > h}} A_m\cdot t \de \mathcal{H}^1 &= \sum_{i = 0}^L \int_{\phi_i}^{\phi_{i+1}} \int_{\cB{\phi > h}} \Curl(A_m) \de x = \\
    &= \sum_{i = 0}^L \rB{\phi_{i+1} - \phi_i} \sum_{j = i}^L b_i = \sum_{i = 1}^L \phi_i b_i.
   \end{split}
  \]
  On the other hand, as a consequence of Lemma~\eqref{lemma:repr_burgers} we have that, for $h \notin \cB{\phi_i}_{i = 1}^L$,
  \[
	  \int_{\partial \cB{\phi > h}}  A_m \cdot t \de \mathcal{H}^1 =  - \int_{\partial \cB{\phi > h}} \rB{\rB{\begin{matrix}x & x^{\perp}\\-x^{\perp}& x\end{matrix}}\cdot \rowvect{\nabla (A_m)^1_1}{\nabla (A_m)^2_2}} t(x)\de \mathcal{H}^1.
  \]
 In particular we see that
 \[
	 b_L = -\sum_{i = 1}^{L-1} \phi_i b_i - \int_0^1 \de h \int_{\partial \cB{\phi > h}} \rB{\rB{\begin{matrix}x & x^{\perp}\\-x^{\perp}& x\end{matrix}}\cdot \rowvect{\nabla (A_m)^1_1}{\nabla (A_m)^2_2}} t(x)\de \mathcal{H}^1.
 \]
 Then, adding $\sum_{i = 1}^{L-1}b_i$ on both sides, we get
 \begin{equation}
   \label{eq:foliation_result}
	 \sum_{i = 1}^L b_i = \sum_{i = 1}^{L-1} \rB{1 - \phi_i} b_i - \int_0^1 \de h \int_{\partial \cB{\phi > h}} \rB{\rB{\begin{matrix}x & x^{\perp}\\-x^{\perp}& x\end{matrix}}\cdot \rowvect{\nabla (A_m)^1_1}{\nabla (A_m)^2_2}} t(x)\de \mathcal{H}^1.
 \end{equation}
 \end{remark}

 In the balls construction we shall need to choose, from a family of balls covering the support of a measure $\mu$, a well disjoint subfamily containing a relevant fraction of the total mass. This is exactly the content of the following Lemma.
\begin{lemma}
   \label{lemma:find_nice_balls}
   Suppose $\bigcup_{i \in I} B(x_i, 30\rho_i) \subset B(0, R)\subset\mathbb{R}^n$ and $\mu$ is a measure on $\mathbb{R}^n$ whose support is contained in $\bigcup_{i \in I} B(x_i, \rho_i)$. Then there exists a subfamily of indices $\tilde{I}\subset I$ and radii $R_i > 3\rho_i$ such that the balls $\cB{B(x_i, 2R_i)}_{i \in \tilde{I}}$ are mutually disjoint, contained in $B(0, R)$ and
   \[\sum_{i \in \tilde{I}} \mu(B(x_i, R_i)) \ge \frac{1}{2\cdot 13^n} \mu\rB{B(0, R)}.\]
  \end{lemma}
  \begin{proof}
	  Let $U_k\coloneqq B_{R\rB{1-2^{-(k+1)}}}\setminus B_{R\rB{1-2^{-k}}}$. If $x_i \in U_k$, then $3\rho_i < \frac{1}{10}2^{-k}R=:r_k$ (since $R\rB{1-2^{-(k+1)}}+ 30\rho_i \le \modulus{x_i}+30\rho_i < R$) and if $\modulus{k - k'} \ge 2$, $x_i \in U_k$ and $x_j \in U_{k'}$, then $B\rB{x_i, r_k}\cap B\rB{x_j, r_{k'}}=\emptyset$. Choose then a maximal family of indices $I_k \subset U_k \cap I$ such that $\modulus{x_i - x_j} \ge \frac{1}{3}r_k$. Then
    \[
     \bigcup_{i \in I}B\rB{x_i, \rho_i} \subset \bigcup_{k\ge 0} \bigcup_{i \in I_k} B\rB{x_i, 2r_k}.
    \]
    Indeed,
    \[
	    \bigcup_{i \in I}B\rB{x_i, \rho_i} = \bigcup_{k\ge 0}\bigcup_{\substack{i\in I\\ x_i \in U_k}} B\rB{x_i, \rho_i}.
    \]
    But
    \[
	    V_k\coloneqq \bigcup_{\substack{i\in I\\ x_i \in U_k}} B\rB{x_i, \rho_i} \subset \bigcup_{i \in I_k} B\rB{x_i, r_k}.
    \]
    For, if $x \in B\rB{x_i, \rho_i}$ for some $x_i \in U_k$ then either $x_i \in I_k$ (and in such a case there is nothing to show) or $x_i\notin I_k$. But in the latter case $\modulus{x_i - x_{\ell}} < \frac{1}{3}r_k$ in order to not contradict the maximality of $I_k$. Hence $x_i \in B\rB{x_{\ell}, \frac{2}{3}r_k}$.
    Clearly, either
    \begin{equation}
	    \label{eq:alternative}
	    \sum_{k \ge 0} \mu\rB{\bigcup_{i \in I_{2k}} B\rB{x_i, r_{2k}}} \ge \frac{1}{2} \mu(B_R) \qquad \text{ or } \qquad \sum_{k \ge 0} \mu\rB{\bigcup_{i \in I_{2k+1}} B\rB{x_i, r_{2k+1}}} \ge \frac{1}{2} \mu(B_R).
    \end{equation}
    If $i, j \in I_k$ and $\modulus{x_i - x_j} \ge 4r_k$, then $B\rB{x_i, \frac{1}{3}r_k} \subset B\rB{x_j, \frac{13}{3}r_k}$, which in turn implies that the balls $\cB{B\rB{x_i, 2r_k}}_{i \in I_k}$ can intersect at most $13^n - 1$ times. Therefore $I_k$ can be split in $N\coloneqq 13^n$ subsets $I_{k, j}$ such that the balls $B\rB{x_i, 2r_k}$ are disjoint.
    Suppose that~\eqref{eq:alternative} holds for even indices (the other case is completely analogous). For every $k\ge 0$, choose a $j(k) \in \cB{1, \cdots, N}$ in guise that
    \[
	    \mu\rB{\bigcup_{i \in I_{2k, j(2k)}} B\rB{x_i, r_{2k}}} \ge \frac{1}{N} \mu \rB{\bigcup_{i \in I_{2k}} B\rB{x_i, r_{2k}}}.
    \]
    Then, since the families $I_{2k, j}$ are disjoint,
    \[
	    \frac{1}{2\cdot 13^n} \mu\rB{B_R(0)} \le \sum_{k \ge 0} \sum_{i \in I_{2k, j(2k)}} \mu\rB{B\rB{x_i, r_{2k}}}.
    \]
    Define then the family of indices
    \[
	    \tilde{I}\coloneqq \cB{i \in I\biggr| \exists k \ge 0\text{ such that } i \in I_{2k, j(2k)}},
    \]
    and the corresponding radii
    \[
	    R_i\coloneqq \max\cB{r_{2k}\biggr| i \in I_{2k, j(2k)}},
    \]
    which are $>0$ for $i \in \tilde{I}$. Then
    \[
	    \frac{1}{2\cdot 13^n} \mu\rB{B_R(0)} \le \sum_{k \ge 0} \sum_{i \in I_{2k, j(2k)}} \mu\rB{B\rB{x_i, r_{2k}}} = \sum_{i \in \tilde{I}} \mu\rB{B\rB{x_i, R_i}}.\qedhere
    \]
  \end{proof}
Henceforth, we deal with competitors of minimizing sequences, that is for every $\epsilon_j \downarrow 0$ and every pair $(A_j, S_j) \in \mathcal{P}\rB{\epsilon, \alpha, L, \lambda, \tau, \ell}$, we can find a competing sequence $(A_j', S_j') \in \mathcal{P}\rB{\epsilon, \alpha, L, \lambda, \tau, \frac{\ell}{2}}$, which we denote again (with an abuse of notation) by $(A_j, S_j)$, which has the properties discussed in Remark~\ref{rmk:competitors}. In particular, each field $A_j$ of such a competing sequence is harmonic outside the singular set $B_{\lambda\epsilon}(S_j)$, and, up to a subsequence, Corollary~\ref{cor:BV} ensures $A_j \to A\in BV(\Omega)$ strongly in $L^2(\Omega)$. Associated to this sequence, we define the measures
\[
 \mu_{1, j}\coloneqq \frac{1}{\tau \epsilon_j} \dist^2(A_j, \so) \mathcal{L}^2\zak \Omega,\qquad \mu_{2, j}\coloneqq \frac{1}{\lambda \epsilon_j} \mathcal{L}^2 \zak B_{\lambda \epsilon}(S_j),\qquad \mu_j\coloneqq \mu_{1, j} + \mu_{2, j},
\]
which, up to subsequences, converge weakly in the sense of measures to $\mu_1$, $\mu_2$ and $\mu$ respectively. We combine together this property and the foliation Lemma~\ref{lemma:foliationR} through a balls construction in the spirit of the one used for the Ginzburg-Landau functional (cf.~\cite{SS} and the references therein and also~\cite{DLGP} for the application of the discrete balls construction to a functional describing systems of dislocations), in order to obtain a density estimate.
  \begin{theorem}[Pseudolinear $1$-density estimate]
  \label{thm:density}
  Let $(A_j, S_j) \in \rB{\epsilon_j, \alpha, L, \tau, \lambda, \ell}$ be a sequence of admissible pairs such that $\mathcal{F}_{\epsilon_j}(A_j, S_j) \le E_{\text{gb}}(\epsilon_j)$, and consider the competing sequence $(A_j', S_j')$ as in Remark~\ref{rmk:competitors}, which (up to a subsequence) converges strongly in $L^2(\Omega)$ to $A \in BV(\Omega)$. 
  There exist constants $C_0 > 0$, $\delta_1\in \rB{0, 1}$ and $\omega_0 > 0$ 
  such that for every $p \in \Omega$ and every $R > 0$ there exists an $\ol{R}\in [R, 2R]$ such that
  \begin{equation}
    \label{eq:density_estimate}
    \modulus{\Curl(A)(B(p, \ol{R})} \le C \omega\rB{\frac{\mu\rB{B(p, 3R)}}{R}} \mu\rB{B(p, R)},
  \end{equation}
  where $\omega: (0, \infty) \to (0, \infty)$ is the continuous increasing function defined as
  \begin{equation}
   \label{eq:189}
   \omega(t)\coloneqq \begin{cases}
   				\omega_0&\text{if }t\ge \delta_1,\\
   				\rB{-\log(t)}^{-1}&\text{if }t<\delta_1.
   			  \end{cases}
  \end{equation}
 \end{theorem}
 \begin{proof}
  We can assume $\mu(B(p, R)) > 0$, otherwise there is nothing to prove. We relabel the competing sequence $(A_j', S_j')$ as $(A_j, S_j)$. Let $\delta_1 > 0$ to be chosen later. If $\mu_2\rB{B(p, 3R)} \ge \delta_1 R$, then by Remark~\ref{rmk:competitors} we have
  \[
	  \mu(B(p, 3R)) \ge C_2 \delta_1 \modulus{\Curl(A)(B(p, R))}.
  \]
  If $\mu_2(B(p, 3R))< \delta_1 R$,
  \[
	  \limsup_{m \to \infty} \mu_{2, \epsilon_m}\rB{B\rB{p, 2R}} \le \mu_2\rB{B\rB{p, 3R}} \le \delta_1 R.
  \]
  Hence, up to a subsequence, which we denote again by $\epsilon_m$, we have that
  \[
	  \modulus{B_{\lambda\epsilon_m}\rB{S_m} \cap B\rB{p, 2R}} \le \lambda \epsilon_m \delta_1 R,\qquad S_m\coloneqq S_{\epsilon_m}.
  \]
  Write $B_{\lambda\epsilon_m}(S_m) = \bigcup_{j \in J_{H, m}} H_{j, m}$, where $H_{j, m}$ are the (closed) connected components of $B_{\lambda\epsilon_m}(S_m)$, and consider only those ones which intersect $B(p, \frac{R}{2})$, that is 
  \[J_{H, m}'\coloneqq \cB{j\in J_{H, m}\biggr| H_{j, m}\cap B\rB{p, \frac{R}{2}}\ne \emptyset}.\]
  Next, cover these components by disjoint balls $\mathcal{B}_0\coloneqq \cB{B\rB{x_{0, i}, \rho_{0,i}}}_{i \in I_0} \equiv \cB{B_{0, i}}_{i \in I_0}$ such that $\sum_{i \in I_0}\rho_{0, i} \le \sum_{j \in J_{H, m}'}\diam(H_{j, m}) \le  \delta_1 R$.  Now, we let these balls grow. Namely, for any positive measure $\mu$ define
  \[
	  \ol{\rho}_{\mu}(x)\coloneqq \sup\cB{\rho > 0\biggr| \mu\rB{B\rB{x, 2\rho}\setminus B\rB{x, \rho}} > \delta_0\rho}.
  \]
  Set $\ol{\rho}_0\coloneqq \ol{\rho}_{\modulus{\nabla \rchi_{\bigcup\mathcal{B}_0 }}}$ and $\ol{\rho}_{i, 0}\coloneqq \ol{\rho}_0(x_{i, 0})$. We can then use Vitali in order to obtain a cover $\cB{B\rB{x_{i, 0}, 6\ol{\rho}_{i, 0}}}_{i \in I_0'}$ such that the balls $B\rB{x_{i, 0}, 2\ol{\rho}_{0, i}}$ are disjoint. Then
  \[
	  6\delta_0 \sum_{i \in I_0'} \ol{\rho}_{i, 0} = 6\sum_{i \in I_0'} \int_{B\rB{x, 2\ol{\rho}_{i, 0}} \setminus B\rB{x, \ol{\rho}_{i, 0}}}\modulus{\nabla \rchi_{\bigcup \mathcal{B}_0}} \le 6\int \modulus{\nabla \rchi_{\bigcup\mathcal{B}_0}} \le 6\sum \rho_{i, 0}.
  \]
  Then, we expand again these balls by a factor of $30$: that is, we consider $\cB{B\rB{x_{i, 0}, 180\ol{\rho}_{i, 0}}}_{i \in I_0'}$. By a merging argument (we refer to \cite{SS} for more details), we get a new family of balls (whose closures are pairwise disjoint) $\mathcal{B}_1\coloneqq \cB{B\rB{x_{i, 1}, \rho_{i, 1}}}_{i \in I_1}$ such that $\sum_{i \in I_1} \rho_{1, i}\le C_0 \sum_{i \in I_0'} \ol{\rho}_{0, i}$, where $C_0\coloneqq \frac{180}{\delta_0}$, which is in turn smaller than $\frac{1}{2}R$, provided $\delta_1$ was chosen small enough. We can then iterate this procedure in order to construct a family of coverings $\cB{\mathcal{B}_k}_{k \ge 0}$, which we can schematize as follows:
  \[
  \begin{split}
   \cdots & \xrightarrow{\text{Merge}} \mathcal{B}_k = \cB{B\rB{x_{k, i}, \rho_{k, i}}}_{i \in I_k} \xrightarrow{\text{Expand}} \cB{B\rB{x_{k, i}, \ol{\rho}_{k, i}}}_{i \in I_k} \xrightarrow{\text{Vitali}} \cB{B\rB{x_{k, i}, 6\ol{\rho}_{k, i}}}_{i \in I_k'} \xrightarrow{30\times} \\
   &\xrightarrow{30\times} \cB{B\rB{x_{k, i}, 180 \ol{\rho}_{k, i}}}_{i \in I_k'} \xrightarrow{\text{Merge}} \mathcal{B}_{k+1} = \cB{B\rB{x_{k+1, i}, \rho_{k+1, i}}}_{i \in I_{k+1}} \xrightarrow{\text{Expand}} \cdots,
   \end{split}
  \]
  where $\ol{\rho}_{k, i}\coloneqq \ol{\rho}_{\modulus{\nabla \rchi_{\bigcup\mathcal{B}_k}}}\rB{x_{k, i}}$. Notice that $\sum_{i \in I_{k+1}}\rho_{k+1, i} \le C_0\sum_{i \in I_k} \rho_{k, i}$, i.e. $\sum_{i \in I_k}\rho_{k, i} \le C_0^k \sum_{i \in I_0} \rho_{0, i}$. Moreover, by construction, each of the balls $B\rB{x_{k, i}, 180\ol{\rho}_{k, i}}$ is contained in precisely one of the $B\rB{x_{k+1, i}, \rho_{k+1, i}}$. That is, we have the inclusions
  \begin{equation}
	  \label{eq:balls_inclusions}
	   B\rB{x_{k+1, i},\rho_{k+1, i}} \supset \bigcup_{j \in I_{k, i}'} B\rB{x_{k, j}, 180\ol{\rho}_{k, j}} \supset 
	   			       \bigcup_{j \in I_{k, i}'} B\rB{x_{k, j}, 6\ol{\rho}_{k, j}} 
				      \supset \spt\rB{\tau_k \zak B\rB{x_{k+1, i}, \rho_{k+1, i}}},
  \end{equation}
  where $\tau_k$ is the measure defined by
  \[
	  \tau_k\coloneqq \sum_{i \in I_k} a_{k, j} \mathcal{L}^2 \zak B_{k, j},\text{where }a_{k, j}\coloneqq \modulus{\fint_{B_{k, j}}\Curl A_{\epsilon_m}\de x}.
  \]
  By Lemma~\ref{lemma:find_nice_balls}, for each $i\in I_{k+1}$ we find a subfamily $I_{k, i}'' \subset I_{k, i}'$ and radii $R_{k, \nu} > 18\ol{\rho}_{k, \nu}$ such that
  \begin{equation}
   \label{eq:good_balls}
   \begin{cases}
     B\rB{x_{k, \nu}, 2R_{k, \nu}} \subset B\rB{x_{k+1, i}, \rho_{k+1, i}}  &\forall \nu \in I_{k, i}'',\\
     B\rB{x_{k, \nu}, 2R_{k, \nu}} \cap B\rB{x_{k, \nu'}, 2R_{k, \nu'}} = \emptyset &\forall \nu \ne \nu',\\
     \sum_{\nu \in I_{k, i}''} \tau_k \rB{B\rB{x_{k, \nu}, R_{k, \nu}}} \ge \frac{1}{C_2} \tau_k\rB{B\rB{x_{k+1, i}, \rho_{k+1, i}}},& C_2\coloneqq 2\cdot\rB{13}^2.
	  \end{cases}
  \end{equation}
  Let
  \[
   K\coloneqq \max\cB{k \ge 1\biggr| \sum_{i \in I_k} \rho_{k, i} < R} + 1.
  \]
  From the discussion above, we have that for a universal constant $c_0 > 0$ (namely, $c_0 = \log\rB{C_0}^{-1}$)
  \[
   K \ge c_0 \log\rB{\frac{R}{\sum_{i \in I_0} \rho_{0, i}}} \ge c_0 \log\rB{\frac{R}{\mu_m(B(p, R))}},\qquad \mu_m\coloneqq \mu_{1, \epsilon_m} + \mu_{2, \epsilon_m}.
  \]
 Arguing as in Remark~\ref{rmk:howto_foliate}, passing to the absolute values in~\eqref{eq:foliation_result} and using the Fleming-Rishel formula, we obtain
 \begin{equation}
	 \label{eq:foliation_consequence}
	 \modulus{\int_{\cB{\phi > 0}} \Curl(A_m) \de x} \le \sum_{B\rB{q_i, r_i} \subset \cB{0 < \phi < 1}} \modulus{\int_{B\rB{q_i, r_i}} \Curl(A_m) \de x} + \int_{\cB{0 < \phi < 1}} \modulus{x} \modulus{\nabla A_{m, \text{sym}}} \modulus{\nabla \phi} \de x.
 \end{equation}
 We can now apply~\eqref{eq:foliation_consequence} to the balls $\tilde{B}_{k, \nu}\coloneqq B\rB{x_{k, \nu}, R_{k, \nu}}$ obtained in~\eqref{eq:good_balls} and the foliation $\phi_{\nu, i}^{(k)}$ given by Lemma~\ref{lemma:foliationR}, for every $k\ge 0$ and $\nu \in I_{k, i}''$, $i \in I_{k+1}$ as in the discussion before. This gives
 \begin{equation}
	 \label{eq:foliation_plus_balls}
	 \begin{split}
		 \sum_{\nu \in I_{k, i}''}\modulus{\int_{\cB{\phi_{\nu, i}^{(k)} > 0}} \Curl(A_m) \de x} \le &\sum_{\nu \in I_{k, i}''}\sum_{B_{k, j} \subset \cB{0 < \phi_{\nu, i}^{(k)} < 1}} \modulus{\int_{B_{k, j}} \Curl(A_m) \de x} +\\
	&+ \sum_{\nu \in I_{k, i}''}\int_{\cB{0 < \phi_{\nu, i}^{(k)} < 1}} \modulus{x} \modulus{\nabla A_{m, \text{sym}}} \modulus{\nabla \phi_{\nu, i}^{(k)}} \de x.
	\end{split}
 \end{equation}
 to both sides of~\eqref{eq:foliation_plus_balls}. Now, for $i \in I_{k+1}$, define the quantities
 \[
	 \begin{split}
		 \mathcal{I}_i^{(k)}&\coloneqq \bigcup_{\nu \in I_{k, i}''} \cB{\phi^{(k)}_{\nu, i} = 1}\text{ (inner balls) },\\
		 \mathcal{A}_i^{(k)}&\coloneqq \bigcup_{\nu \in I_{k, i}''} \cB{0 < \phi^{(k)}_{\nu, i} < 1}\text{ (annuli) },\\
		 \mathcal{R}_i^{(k)}&\coloneqq \bigcup_{B_{k, j}\subset B_{k+1, i} \setminus \bigcup_{\nu \in I_{k, i}''} \cB{\phi_{\nu, i}^{(k)} > 0}} B_{k, j} \text{ (remaining balls) },\\
		 \mathcal{J}_i^{(k)}&\coloneqq \sum_{\nu \in I_{k, i}''}\int_{\cB{0 < \phi_{\nu, i}^{(k)} < 1}} \modulus{x} \modulus{\nabla A_{\epsilon, \text{sym}}} \modulus{\nabla \phi_{\nu, i}^{(k)}} \de x.
	 \end{split}
 \] 
 Since the balls $\cB{\tilde{B}_{k, \nu}}_{\nu \in I_{k, i}''}$ were given by Lemma~\ref{lemma:find_nice_balls} we have
 \begin{equation}
	 \label{eq:gain}
	 \tau_k\rB{\mathcal{A}_i^{(k)}} =   \tau_k\rB{B_{k+1, i}} - \tau_k\rB{\mathcal{I}_i^{(k)}} - \tau_k\rB{\mathcal{R}_i^{(k)}} \le \rB{1 - \epsilon_0}  \tau_k\rB{B_{k+1, i}} - \tau_k\rB{\mathcal{R}_i^{(k)}},
 \end{equation}
 where $\epsilon_0\coloneqq C_2^{-1} = \rB{2\rB{13}^2}^{-1}<1$. We then add the term
 \[
	 P_{k, i}\coloneqq \modulus{\sum_{B_{k, j}\subset B_{k+1, i} \setminus \bigcup_{\nu \in I_{k, i}''} \cB{\phi_{\nu, i}^{(k)} > 0}} \int_{B_{k, j}} \Curl(A_{\epsilon_m})\de x}
 \]
 to both sides of~\eqref{eq:foliation_plus_balls}, which gives, using~\eqref{eq:gain},
 \begin{equation}
	 \label{eq:est_B_i}
	\begin{split}
	 \modulus{\int_{B_{k+1, i}} \Curl(A_{\epsilon_m})\de x} & \le P_{k, i} + \sum_{\nu \in I_{k, i}''}\modulus{\int_{\cB{\phi_{\nu, i}^{(k)} > 0}} \Curl(A_{\epsilon_m}) \de x} \le \\
	 &\le \sum_{\nu \in I_{k, i}''}\sum_{B_{k, j} \subset \cB{0 < \phi_{\nu, i}^{(k)} < 1}} \modulus{\int_{B_{k, j}} \Curl(A_{\epsilon_m}) \de x} + P_{k, i} + \mathcal{J}_i^{(k)}\le \\
	 &\le \tau_k\rB{\mathcal{A}_i^{(k)}} + P_{k, i} + \mathcal{J}_i^{(k)} \le \rB{1 - \epsilon_0} \tau_k\rB{B_{k+1, i}} + \mathcal{J}_i^{(k)}.
	\end{split}
 \end{equation}
 We then just need to sum up~\eqref{eq:est_B_i} for $i \in I_{k+1}$ in order to get
 \begin{equation}
	 \label{eq:est_tau_k_noiter}
	 \tau_{k+1}\rB{B\rB{p, 2R}} \le \rB{1-\epsilon_0} \tau_k\rB{B\rB{p, 2R}} + \mathcal{J}^{(k)},\qquad \mathcal{J}^{(k)}\coloneqq \sum_{i \in I_{k+1}} \mathcal{J}_i^{(k)},
 \end{equation}
which immediately implies
\footnote{
 Recall that, using a Whitney covering, one can prove the existence of a constant $c = c_n > 0$ such that for every harmonic function $u$ in an open set $U\subset \mathbb{R}^n$
  \begin{equation}
   \label{eq:whitney}
   \int_U \modulus{\nabla u}^2 \dist^2(x, \partial U) \de x \le c_n \int_U \modulus{u}^2 \de x.
  \end{equation}
}
\begin{equation}
  \label{eq:blabla}
  \tau_{k+1}(B(p, 2R)) \le \rB{1 - \epsilon_0}\tau_k(B(p, 2R)) + C_{\epsilon_0} \mu_{1, \epsilon_m}(\mathcal{A}(k)) + \frac{\epsilon_0}{4}(1 + n_k)\tau \epsilon_m,
\end{equation}
where
\[
 \mathcal{A}(k):=\bigcup_{i \in I_{k+1}} \mathcal{A}^{(k)}_i,\text{ and } n_k:=\sum_{i \in I_k} \sum_{\nu \in \tilde{I}_{k, i}} \#\cB{B^{(k)}_{\ell} \subset \cB{0 < \varphi^{(k)}_{\nu, i} < 1}}.
\]
Now, either
\begin{equation}
  \label{eq:alt1tau}
  \tau_{k+1}(B(p, 2R)) < n_k \tau \epsilon_m,
\end{equation}
or
\begin{equation}
  \label{eq:alt2tau}
  \tau_{k+1}(B(p, 2R)) \ge n_k \tau \epsilon_m.
\end{equation}
If~\eqref{eq:alt2tau} holds, we have two possible subcases: either $\tau_{k+1}(B(p, 2R)) = 0$, or $\tau_{k+1}(B(p, 2R)) > 0$. In the latter case, because of the first quantization of the Burgers vector, we have that $\tau_{k+1}(B(p, 2R)) \ge \tau \epsilon_m$. Thus
\begin{equation}
 \label{eq:consequence_alt2tau}
 \tau_{k+1}(B(p, 2R)) \ge \frac{1}{2}(1 + n_k)\epsilon_m.
\end{equation}
In particular, when~\eqref{eq:alt2tau} holds, we have from~\eqref{eq:blabla} that
\begin{equation}
 \label{eq:blabla2}
 \tau_{k+1}(B(p, 2R)) \le \rB{1 - \frac{\epsilon_0}{2}} \tau_k(B(p, 2R)) + C_{\epsilon_m} \mu_{1, \epsilon_m}(\mathcal{A}(k)).
\end{equation}
We then add~\eqref{eq:blabla2} and~\eqref{eq:alt1tau}, in order to find
\begin{equation}
 \label{eq:blabla3}
 \tau_{k+1}(B(p, 2R)) \le \rB{1 - \frac{\epsilon_0}{2}} \tau_k(B(p, 2R)) + C_{\epsilon_0} \mu_{1, \epsilon_m}(\mathcal{A}(k)) + n_k \tau \epsilon_m.
\end{equation}
An iteration gives then immediately
\begin{equation}
 \label{eq:blabla4}
 \begin{split}
   \tau_{k+1}(B(p, 2R)) &\le \rB{1 - \frac{\epsilon_0}{2}}^k \tau_0(B(p, 2R)) + C_{\epsilon_0} \sum_{\ell = 0}^k \rB{1-\frac{\epsilon_0}{2}}^{\ell}\mu_{1, \epsilon_m}(\mathcal{A}(k - \ell)) +\\
   &+ \tau \epsilon_m \sum_{\ell = 0}^k \rB{1-\frac{\epsilon_0}{2}}^{\ell} n_{k - \ell}.
 \end{split}
\end{equation}
Now, we sum up~\eqref{eq:blabla4} for $k\in {1, \cdots, K-1}$. Using the fact that $\tkpr{k}$ decreases as $k$ increases, we find
\begin{equation}
  \label{eq:blabla5}
  K \tau_K(B(p, 2R)) \le C\cB{\mu_{m}(B(p, 2R)) + \tau\epsilon \sum_{k = 1}^{K - 1} n_k}.
\end{equation}
On the other hand, by construction, all the annuli and balls are disjoint, which means
\[
  \sum_{k = 1}^{K - 1} n_k \le N_0,
\]
where $N_0$ is the number of connected components of $B_{\lambda \epsilon_m}(S_{\epsilon_m})$ contained in $B(p, 3R)$, which can be in turn be estimated by
 \[
  N_0 \le C \frac{\modulus{B_{\lambda \epsilon_m}(S_{\epsilon_m}) \cap B\rB{p, 2R}}}{\rB{\lambda \epsilon_m}^2} = C\frac{\mu_{2, \epsilon_m}\rB{B(p, 2R)}}{\epsilon_m}.
 \]
In particular, we obtain
\begin{equation}
 \label{eq:gain2}
 \tau_K(B(p, 2R)) \modulus{\log\rB{\frac{\mu_m(B(p, 2R))}{R}}} \le C_{\tau} \mu_m\rB{B(p, 2R)}.
\end{equation}
 Now, since $\modulus{B_{\lambda \epsilon_m}(S_m) \cap B\rB{p, 3R}} \le \lambda \epsilon_m \delta_1 R$, we can find an $R(m) \in [R, 2R]$ such that 
 \[
  \tau_K \rB{B(p, 2R)} \ge \modulus{\int_{B\rB{p, R(m)}} \Curl A_m \de x}.
 \]
 Up to a subsequence, we can always assume that $R(m)\to \ol{R} \in [R, 2R]$. Moreover, since $\cB{\Curl A_{\epsilon_m}}$ quasi-converges to $(\Curl A, \xi)$, with $\xi(\Omega) < \infty$, we can also assume $\xi\rB{\partial B(p, \ol{R})} = 0$ (up to increasing the constant $C_{\lambda}$ in the right hand side of~\eqref{eq:gain2} by a factor of 2). In particular, we have
 \[
  \limsup_{m \to \infty } \modulus{\int_{B(p, R(m))} \Curl A_{m} \de x} \ge \modulus{\int_{B(p, \ol{R})} \de \Curl A}.
 \]
 Taking the limit superior as $m \to \infty$ in~\eqref{eq:gain2}, we find
 \begin{equation}
  \label{eq:gain3}
  \modulus{\int_{B(p, \ol{R})} \de \Curl(A)} \modulus{\log\rB{\frac{\mu(B\rB{p, 3R})}{R}}} \le C_0 \mu\rB{B(p, 3R)}.
 \end{equation}
 In particular, we can choose $\omega$ as in~\eqref{eq:189} and obtain~\eqref{eq:density_estimate}.
\end{proof}

Theorem~\ref{thm:density} is giving an estimate of the norm of $\Curl(A)$ on balls, while in order to obtain an estimate for the derivative $DA$ we would need (by virtue of Proposition~\ref{prop:BV}) an upper bound on the total variation of $\Curl(A)$. The key observation in order to prove such an estimate is that, by the definition of supremum limit, we are allowed to take a covering with balls of the same radii.
	\begin{lemma}
	 \label{lemma:vmeas}
	 Let $T$ be a vector valued Radon measure and $\mu$ be a positive finite Radon measure, both defined on $\mathbb{R}^n$. Suppose that there exists a constant $C_0 > 1$ such that for every $x \in \Omega$ and every $R > 0$
	 \begin{equation}
	 	\label{eq:hp_density}
	 	\modulus{T\rB{B(x, R)}} \le \omega\rB{\frac{\mu\rB{B(x, C_0 R)}}{R^{\beta}}} \mu\rB{B(x, C_0 R)},
	 \end{equation}
	\end{lemma}
	where $\beta \in \cB{1, \cdots, n-1}$ and $\omega: (0, \infty) \to (0, \infty)$ is an increasing function such that $\omega(\delta)\to 0$ as $\delta \to 0$. Then
	\begin{enumerate}[(a)]
	 \item $\modulus{T}\rB{\Omega \setminus S} = 0$, where 
	 \[
	 	S\coloneqq \cB{x \in \Omega \biggr| \Theta^*(x) > 0},\qquad \Theta^*(x)\coloneqq \Theta^*_{\beta}(\mu, x)\coloneqq \limsup_{R\downarrow 0} \frac{\mu\rB{B(x, R)}}{R^{\beta}};
	 \]
	 \item $\mathcal{H}^{\beta} \zak S$ is $\sigma$-finite;
	 \item $\modulus{T} \le C_n \rB{\omega\circ \Theta^*} \mu\zak S$, where $C_n > 0$ is a constant depending only on the dimension.
	\end{enumerate}
	In particular, if $T = DA$ for some $A \in BV(\Omega)^n$, then $DA = D^J A = \modulus{A^+ - A^-}\otimes \nu_A \mathcal{H}^{n-1}\zak S_A$ and
	 \begin{equation}
	  \label{eq:final}
	  g^{-1}\rB{\modulus{A^+ - A^-}} \mathcal{H}^{n-1}\zak S_A \le C\mu,
	 \end{equation}
	 where $g(t)\coloneqq t\omega(t)$.
	\begin{proof}
	 From the definition of limit superior,
	 \[
	  G_s\coloneqq \cB{x \in \Omega\biggr| \Theta^*(x) \le s} \subset \bigcap_{\delta > 0} \bigcup_{R > 0} G_{s, R, \delta},
	 \]
	 where
	 \[
	  G_{s, R, \delta}\coloneqq \cB{x \in \Omega\biggr| \frac{\mu\rB{B(x, C_0 \rho)}}{\rho^{\beta}}< s + \delta \quad \forall \rho < R}.
	 \]
	 For any $\mu$-measurable set $E$, consider the $r$-tubular neighborhood $U_r = B_r\rB{E\cap G_{s, R, \delta}}$. Fix $\rho < \min\cB{R, \frac{r}{C_0}}$. We can find $K = K(n)$ (depending only on the dimension $n$) disjoint families of balls balls $\mathcal{B}_k\coloneqq \cB{B_i^{(k)}}_{i \in I_k}\equiv\cB{B\rB{x_i^{(k)}, \rho}}_{i \in I_k}$, $k = 1, \cdots, K$ whose union covers $E_{s, R, \delta}\coloneqq E\cap G_{s, R, \delta}$, that is
	 \[
	  E_{s, R, \delta} \subset \bigcup_{k = 1}^K \bigcup_{i \in I_k} B_i^{(k)},\qquad B_i^{(k)} \cap B_j^{(k)} = \emptyset \forall i\ne j.
	 \]
	 Moreover, the choice of $\rho$ ensures
	 \[
	  \frac{\mu\rB{B(x_i^{(k)}, C_0 \rho)}}{\rho^{\beta}} < s + \delta,\quad \forall i\in I_k,\quad k\in \cB{1, \cdots K}
	 \]
	 and
	 \[
	  C_0 B_i^{(k)} \subset U_r.
	 \]
	 Let $f\coloneqq \frac{\de T}{\de \modulus{T}}$, $\modulus{f} = 1$ $\modulus{T}$-a.e., and $\phi \in \mathcal{C}_c(\Omega)$. Then, using~\eqref{eq:hp_density},
	 \[
	 	  \begin{split}
		  \modulus{T}\rB{E_{s, R, \delta}} &\le \sum_{k = 1}^K \sum_{i \in I_k} \modulus{\int_{B_i^{(k)}} \scal{f}{\de T}} \le \\
	 &\le \sum_{k = 1}^K \sum_{i \in I_k} \left\{\int_{B_i^{(k)}} \modulus{f - \phi} \de \modulus{T} + \int_{B_i^{(k)}} \modulus{\phi(x) - \phi\rB{x_i^{(k)}}} \de \modulus{T}(x) \right. +\\
	 &\qquad\qquad\quad+ \left. \modulus{\scal{\phi\rB{x_i^{(k)}}}{T\rB{B_i^{(k)}}}}\right\} \le \\
	 &\le K\cB{\int_{U_r} \modulus{f - \phi} \de \modulus{T} + \rB{\sup_{\modulus{x - y} < \rho} \modulus{\phi(x) - \phi(y)}} \modulus{T}\rB{U_r}} +\\
	 &\quad+ C_n \norm{\phi}_{\infty}\omega\rB{\delta + s}\mu(U_r).
	  \end{split}
	 \]
	 Define
	 \[
	  (I)\coloneqq \int_{U_r} \modulus{f - \phi} \de \modulus{T},\qquad (II)\coloneqq \rB{\sup_{\modulus{x - y} < \rho} \modulus{\phi(x) - \phi(y)}} \modulus{T}\rB{U_r}.
	 \]
	 As $\rho \to 0$, we see that $(II) \to 0$, while if we consider a sequence of functions $\phi$ converging to $f$, also $(I)\to 0$. As $G_{s, R, \delta}$ is increasing in $R$, taking $R\to \infty$ we can replace $E_{s, R, \delta}$ on the left hand side with the union $E_{s, \delta}\coloneqq \bigcup_{R > 0} E_{s, R, \delta}$. Since this holds for every $\delta > 0$, we can let $\delta \to 0$ and recover $E_s = E \cap \cB{\Theta^* > s}$ on the left hand side. Finally, taking $r \to 0$ and using the fact that $\mu$ is a Radon measure, we find that for every $\mu$-measurable set $E$
	 \begin{equation}
	  \label{eq:666}
	  \modulus{T}(E \cap G_{s}) \le \omega(s) \mu\rB{E \cap G_{s}}.
	 \end{equation}
	 Since $\omega(s)\to 0$ as $s\to 0$, we have that
	 \[
	  \modulus{T}(\Omega\setminus S) = 0,
	  	 \]
	 i.e. (a). Set $S_{\delta}\coloneqq \cB{x \biggr| \Theta^*(x) > \delta}$. Then clearly $\mathcal{H}^{\beta}(S_{\delta}) \le C_n\frac{1}{\delta}\mu(\mathbb{R}^n) < \infty$. In particular, $\mathcal{H}^{\beta} \zak S$ is $\sigma$-finite, thus (b) is proven.\\
	 Now, for every $\zeta > 0$, we can find a compact set $H = H(\zeta)$ such that $\Theta^*|_H$ is continuous and $\modulus{T}(\Omega\setminus H)<\zeta$. For $\eta > 0$, let $\Phi^*(\eta) > 0$ be such that
	 \[
	  x, y \in H, \modulus{x - y} \le \Phi^*(\eta) \Longrightarrow \modulus{\Theta^*(x) - \Theta^*(y)} \le \eta.
	 \]
	 Consider a sequence $\cB{a_i}_{i\ge 1}$ such that $(0, \infty) = \bigcup_{i \ge 1}(a_i, a_{i+1}]$ and $\modulus{a_{i+1} - a_i} < \Phi^*(\eta)$. For any Borel set $F$, let
	 \[
	  F_i\coloneqq F\cap \cB{x\biggr| \Theta^*(x) \in (a_i, a_{i+1}]}.
	 \]
	 Let $\mu_0\coloneqq \mu\zak S$. Using~\eqref{eq:666} with $E = F \cap \cB{x\biggr| \Theta^*(x) > a_i}$,
	 \[
	  \begin{split}
	   \modulus{T}(F) &\le \zeta + \modulus{T}(F \cap H) \le \zeta + \sum_{i\ge 1}\modulus{T}(F_i\cap H) \le\\
	   				  &\le \zeta + \sum_{i\ge 1}\modulus{T}\rB{F \cap H \cap \cB{x\biggr| \Theta^*(x) > a_i} \cap G_{a_{i+1}}} \le \\
	   				  &\le \zeta + C_n \sum_{i \ge 1}\omega\rB{a_{i + 1}} \mu_0(F_i\cap H) = \zeta + C_n \sum_{i \ge 1}\int_{F_i\cap H} \omega\rB{a_{i + 1}} \de \mu_0 \le \\
	   				  &\le \zeta + C_n \sum_{i \ge 1} \int_{F_i\cap H} \omega\rB{\Theta^*(x)} \de \mu_0 + C_n \sum_{i \ge 1} \int_{F_i\cap H} \modulus{\omega\rB{\Theta^*(x)} - \omega\rB{a_{i+1}}} \de \mu_0 \le \\
	   				  &\le \zeta + C_n \sum_{i \ge 1} \int_{F_i\cap H} \omega\rB{\Theta^*(x)} \de \mu_0 + C_n \eta \mu_0(F\cap H) =\\
	   				  &= \zeta + C_n \int_{F\cap H} \omega\rB{\Theta^*(x)} \de \mu_0 + C_n \eta \mu_0(F\cap H ).
	  \end{split}
	 \]
	 By the arbitrariness of $\zeta$, $\eta$ and the Borel set $F$, we infer that
	 \[
	  \modulus{T} \le C_n \rB{\omega\circ \Theta^*} \mu\zak S,
	 \]
	 i.e. (c). Now, suppose $T = DA$ for some $A \in BV(\Omega)^m$. Then from (c), we see that $DA = \modulus{A^+(x) - A^-(x)} \otimes \nu_A \mathcal{H}^{\beta}\zak (S\cap S_A)$, where $\beta\coloneqq n-1$. Our first claim is that
	 \[
	  \modulus{A^+(x) - A^-(x)} \le C\Theta^*(x)\omega\rB{\Theta^*(x)}\qquad \text{ for }\mathcal{H}^{\beta}-{\text{a.e. }} x\in S\cap S_A.
	 \]
	 Let $E\subset \mathbb{R}^n$ be a Borel set. For any $\zeta > 0$, we can find $H = H(\zeta)$ compact such that $\Theta^*|_H$ is continuous and $\mu\rB{\mathbb{R}^n \setminus H} \le \zeta$. Since $S$ is rectifiable, we can assume without loss of generality that the $\beta$-density of each $x \in S \cap H \cap E$ is $1$, namely
  \[
   \lim_{\rho \downarrow 0} \frac{\mathcal{H}^{\beta}\rB{S\cap H \cap E\cap B\rB{x, \rho}}}{c_{\beta}\rho^{\beta}} = 1,
  \]
  where $c_{\beta} > 0$ is a constant dependent only on $\beta > 0$. From this and the definition of limit superior, for every $\eta > 0$, $k \in\mathbb{N}$ and $x \in E\cap S_{\xi} \cap H=:G_{\xi}$, $\xi>0$, we can find a radius $\rho_k(x)\le k^{-1}$ such that, for a constant $C = C(\beta) > 0$,
  \begin{equation}
   \label{eq:rho}
   \begin{cases}
   	C\rB{1-\eta}\rho_k(x)^{\beta} \le \mathcal{H}^{\beta}\rB{G_{\xi}\cap \ol{B\rB{x, \rho_k(x)}}} \le C(1+\eta)\rho_k(x)^{\beta},\\
   	\Theta^*(x) \ge \frac{\mu\rB{B(x, \rho_k(x))}}{\rho_k(x)^{\beta}} - \eta.
   \end{cases}
  \end{equation}
  We then consider, for $N>1$, the fine cover of $G_{\xi}$
  \[
   \mathcal{F}_N\coloneqq \cB{\ol{B\rB{x, \rho_k(x)}}\biggr| x \in G_{\xi},\quad k\ge N}.
  \]
  from which, by Vitali-Besicovitch Theorem, we can extract a disjoint family $\mathcal{F}'_N = \cB{B(x_i, \rho_i)}_{i \ge 1}$ such that
  \[
   \mu\rB{G_{\xi}\setminus \bigcup\mathcal{F}'_N} = 0.
  \]
  Then
  \[
   \begin{split}
    \int_{E \cap S_{\xi}} \omega\rB{\Theta^*(x)} \de \mu(x) &\le C\zeta + \int_G  \omega\rB{\Theta^*(x)} \de \mu(x) = \\
    			&= C\zeta + \sum_i \int_{\ol{B\rB{x_i, \rho_i}}\cap G_{\xi}} \omega\rB{\Theta^*(x)} \de \mu (x) \le\\
    			&\le C\zeta + \sum_i \omega\rB{\Theta^*(x_i)} \mu\rB{\ol{B\rB{x_i, \rho_i}} \cap G} +\\
    			&\quad + \rB{\sup_{\substack{x, y \in G\\ \modulus{x - y} \le N^{-1}}}\modulus{\omega\rB{\Theta^*(x)} - \omega\rB{\Theta^*(y)}}} \mu(G)\le\\
    			&\le C\zeta + \sum_i \omega\rB{\Theta^*(x_i)} \rho_i^{\beta}\rB{\eta + \Theta^*(x_i)} + o_N(1) \le\\
    			&\le C\zeta + \sum_i \Theta^*(x_i) \omega\rB{\Theta^*(x_i)} \rho_i^{\beta} + \eta\sum_i \omega\rB{\Theta^*(x_i)} \rho_i^{\beta} + o_N(1).
   \end{split}
  \]

  Using~\eqref{eq:rho}, we find (setting $g(s)\coloneqq s\omega(s)$ and $\tilde{g}\coloneqq g\circ \Theta^*$),
  \[
   \begin{split}
   \sum_i \tg(x_i) \rho_i^{\beta} &\le \frac{C}{1-\eta} \sum_i \tg(x_i) \mathcal{H}^{\beta}(G_{\xi}\cap \ol{B\rB{x_i, \rho_i}}) \le \\
   &\le \int_{S\cap E} \tg(y) \de \mathcal{H}^{\beta}(y) + \rB{\sup_{\substack{x, y \in G,\\\modulus{x-y}\le N^{-1}}} \modulus{\tg(x) - \tg(y)}} \mathcal{H}^{\beta}\rB{S_{\xi}} \le \\
   &\le \int_{S\cap E} \tg(y) \de \mathcal{H}^{\beta}(y) + o_N(1) \frac{\mu(\mathbb{R}^n)}{\xi}.
   \end{split}
  \]
  and, since $\omega$ is bounded,
  \[
   \begin{split}
    \eta \sum_i \omega\rB{\Theta^*(x_i)} \rho^{\beta} \le C\eta \norm{\omega}_{\infty} \frac{\mathcal{H}^{\beta}(S_{\xi})}{1-\eta}.
   \end{split}
  \]
  That is,
  \begin{equation}
   \label{eq:fred}
   \int_{E\cap S_{\xi}} \omega\rB{\Theta^*(x)} \de \mu(x) \le C\zeta + o_N(1)\frac{1}{\xi} + C\eta\frac{\norm{\omega_{\infty}}}{\xi(1-\eta)} + \int_{S\cap E} \tg(x) \de \mathcal{H}^{\beta}(x).
  \end{equation}
  Then, in~\eqref{eq:fred} we first let $N\to \infty$, then $\zeta \to 0$ and $\eta\to 0$. By the arbitrariness of $\xi > 0$ and the set $E$, we finally get
  \[
   \rB{\omega\circ\Theta^*} \mu\zak\cB{x\in S\biggr| \Theta_1(S, x) = 1} \le \tg \mathcal{H}^{\beta}\zak \cB{x\in S\biggr| \Theta_1(S, x) = 1}.
  \]
  That is, since $S$ is rectifiable,
  \begin{equation}
   \label{eq:est11}
   \modulus{A^+(x) - A^-(x)}\le C\Theta^*(x)\omega\rB{\Theta^*(x)},\qquad \text{ for }\mathcal{H}^{\beta}-{\text{a.e. }}x\in S.
  \end{equation}
  We rewrite~\eqref{eq:est11} as
  \begin{equation}
   \label{eq:est12}
   f\rB{\modulus{A^+(x) - A^-(x)}}\le C\Theta^*(x),\qquad \text{ for }\mathcal{H}^{\beta}-{\text{a.e. }}x\in S.
  \end{equation}
  where $f\coloneqq g^{-1}$. We proceed now with the proof of the second step. Let $E\subset\mathbb{R}^n$ Borel and $\xi > 0$. We re-define $G_{\xi}$ as 
  \[G_{\xi}\coloneqq E\cap \cB{x \in S\biggr|\Theta_1(S, x) = 1\text{ and } \Theta^*(x) > \xi}.\]
  For every $\eta > 0$ and $k \in \mathbb{N}$, we can find a $\rho_k(x) \le k^{-1}$ such that
  \begin{equation}
   \label{eq:rho2}
   \begin{cases}
   	C\rB{1-\eta}\rho_k(x)^{\beta} \le \mathcal{H}^{\beta}\rB{G_{\xi}\cap \ol{B\rB{x, \rho_k(x)}}} \le C(1+\eta)\rho_k(x)^{\beta},\\
   	\Theta^*(x) \le \frac{\mu\rB{B(x, \rho_k(x))}}{\rho_k(x)^{\beta}} + \eta,\\
   	A^+(y) - A^-(y) = A^+(x) - A^-(x),\qquad \forall y \in \ol{B\rB{x, \rho_k(x)}} \cap G_{\xi},\\
   	f\rB{\modulus{A^+(x) - A^-(x)}}\le C\Theta^*(x)\qquad \forall x \in G_{\xi}.
   \end{cases}
  \end{equation}
  As before, for $N>1$, we define the fine cover
  \[
   \mathcal{F}_N\coloneqq \cB{\ol{B\rB{x, \rho_k(x)}}\biggr| x \in G_{\xi},\quad k\ge N},
  \]
  from which we extract a disjoint family $\mathcal{F}'_N = \cB{B\rB{x_i, \rho_i}}_{i \ge 1}$ such that
  \[
   \mathcal{H}^{\beta}\rB{G_{\xi}\setminus \bigcup \mathcal{F}'_N} = 0.
  \]
  We have
  \[
   \begin{split}
    \int_{E\cap S_{\xi}} f\rB{\modulus{A^+(x) - A^-(x)}} \de \mathcal{H}^{\beta} &= \int_{G_{\xi}} f\rB{\modulus{A^+(x) - A^-(x)}} \de \mathcal{H}^{\beta} =\\
    &= \sum_i \int_{G_{\xi} \cap \ol{B\rB{x_i, \rho_i}}} f\rB{\modulus{A^+(x) - A^-(x)}} \de \mathcal{H}^{\beta}  = \\
    &= \sum_i f\rB{\modulus{A^+(x_i) - A^-(x_i)}}\mathcal{H}^{\beta}(G_{\xi} \cap \ol{B\rB{x_i, \rho_i}}) \le\\
    &\le C \rB{1+\eta}\sum_i \Theta^*(x_i) \rho_i^{\beta} \le\\
    &\le C \rB{1+\eta}\sum_i \rB{\frac{\mu\rB{B(x, \rho_i)}}{\rho_i^{\beta}} + \eta}\rho_i^{\beta} \le\\
    &\le C \rB{1+\eta} \mu\rB{B_{\frac{1}{N}}\rB{G_{\xi}}} + C\eta\frac{1+\eta}{1-\eta} \mathcal{H}^{\beta}(S_{\xi}).
   \end{split}
  \]
  As  $N \to \infty$ and $\eta \to 0$, by the arbitrariness of $\xi > 0$ and $E$ we have
  \[
   f\rB{\modulus{A^+ - A^-}} \mathcal{H}^{\beta}\zak S \le C \mu.\qedhere
  \]
	\end{proof}
	From Theorem~\ref{thm:density}, Lemma~\ref{lemma:vmeas},\cite[Proposition 1]{LL} and a slicing argument, we immediately infer the following
\begin{corollary}
 \label{cor:microrot}
 There exist a constant $C > 0$ and $\alpha_0 > 0$ such that for every $0 < \alpha \le \alpha_0$ and every sequence of pairs $(A_j, S_j) \in \mathcal{P}\rB{\epsilon_j, \alpha, L, \tau, \lambda, \ell}$, $\epsilon_j \to 0$ with $\mathcal{F}_{\epsilon_j}(A_j, S_j) \le E_{\text{gb}}(\epsilon_j)$, there exists another sequence $(A_j', S_j') \in \mathcal{P}\rB{\epsilon_j, \alpha, L, \tau, \lambda, \frac{\ell}{2}}$ such that $\mathcal{F}_{\epsilon_j}(A_j', S_j') \le C\mathcal{F}_{\epsilon_j}(A_j', S_j')$ which, up to a subsequence, converges strongly in $L^2(\Omega)$ to a microrotation $A$ and
 \[
  \modulus{A^+ - A^-}\modulus{\log(\modulus{A^+ - A^-})} \mathcal{H}^1\zak S_A \le C \mu,
 \]
 where $\mu$ is the weak$^*$ limit of the measures
 \[
  \mu_j \coloneqq  \frac{1}{\tau \epsilon_j} \dist^2(A_j', \so)\mathcal{L}^2\zak \Omega + \frac{1}{\lambda\epsilon_j} \mathcal{L}^2 \zak S_j'.
 \]
 In particular,
 \[
  C\alpha L \epsilon_j\modulus{\log(\alpha)} \le \mathcal{F}_{\epsilon_j}(A_j, S_j).
 \]
\end{corollary}

\begin{remark}
	Without imposing the first quantization of the Burgers vector, one could obtain another estimate involving the square root of the logarithm instead, which would be again optimal in such a different context.
\end{remark}

\end{document}